\documentclass[11pt]{article}
\pdfoutput=1 

\usepackage{lmodern}
\usepackage{booktabs}
\usepackage[margin=1in]{geometry}
\usepackage{amsfonts,amssymb,amsmath,amsthm,mathtools,thmtools}
\usepackage{microtype,xspace,bm,floatrow}
\usepackage{braket}
\usepackage{comment}
\usepackage{mdframed}

\usepackage{doi}
\usepackage{hyperref}
\usepackage[dvipsnames]{xcolor}
\hypersetup{
	colorlinks=true,
	urlcolor=MidnightBlue,
	linkcolor=MidnightBlue,
	citecolor=MidnightBlue,
	unicode
}

\newcommand{\negl}{\mathsf{negl}}

\newcommand{\defeq}{:=}

\newcommand{\calX}{\mathcal{X}}

\newcommand{\FF}{\mathbb{F}}
\newcommand{\vv}{{v}}

\newcommand{\vx}{{x}}

\newcommand{\vz}{{z}}
\newcommand{\ve}{{e}}

\newcommand{\GRS}{\mathrm{GRS}}

\newcommand{\hw}{\mathsf{hw}}
\newcommand{\hwUF}{\mathsf{hw}}

\newcommand{\bad}{\mathsf{BAD}}
\newcommand{\good}{\mathsf{GOOD}}
\newcommand{\Tr}{\mathrm{Tr}}
\newcommand{\mfol}{(m)}
\newcommand{\listdecode}{\mathsf{ListDecode}}
\newcommand{\decode}{\mathsf{Decode}}

\newcommand{\RS}{\mathrm{RS}}
\newcommand{\vzero}{0}

\newcommand{\gooderrors}{\mathcal{G}}
\newcommand{\baderrors}{\mathcal{B}}
\newcommand{\dist}{\mathcal{D}}

\newcommand{\QFT}{\mathsf{QFT}}

\usepackage{enumitem}
\setlist{  
  listparindent=\parindent,
  parsep=0pt,
}

\usepackage[capitalise,nameinlink,noabbrev]{cleveref}
\crefname{equation}{}{}

\usepackage{algorithm}
\usepackage[noend]{algpseudocode}

\usepackage{tikz}
\usepackage{graphicx}
\usetikzlibrary{positioning,arrows.meta,math,shapes.geometric,decorations.pathmorphing}

\usepackage[font=small,labelfont=bf]{caption} 
\usepackage{subcaption}

\usepackage[english]{babel}
\addto\extrasenglish{}
\addto\extrasenglish{}

\usepackage{lpic}

\makeatletter
\DeclareRobustCommand\bfseries{%
  \not@math@alphabet\bfseries\mathbf
  \fontseries\bfdefault\selectfont\boldmath}
\makeatother

\declaretheorem[name=Theorem]{theorem}

\declaretheorem[name=Lemma,sibling=theorem]{lemma}

\declaretheorem[name=Claim,sibling=theorem]{claim}
\declaretheorem[name=Definition, sibling=theorem, style=definition]{definition}
\declaretheorem[name=Fact,sibling=theorem]{fact}

\DeclareMathOperator{\poly}{poly}
\DeclareMathOperator{\codim}{codim}

\newcommand{\newclass}[2]{\newcommand{#1}{{\text{\upshape\sffamily #2}}\xspace}}

\newclass{\NP}{NP}
\newclass{\cA}{A}
\newclass{\cB}{B}
\newclass{\BPP}{BPP}
\newclass{\BQP}{BQP}
\newclass{\TFNP}{TFNP}

\newcommand{\newprob}[2]{\newcommand{#1}{{\text{\upshape\scshape #2}}\xspace}}
\newprob{\NC}{NullCodeword}
\newprob{\NCshort}{NC}
\newprob{\bNC}{BiNC}
\newprob{\TbNC}{T-BiNC}
\newprob{\YZ}{NC}
\newprob{\TYZ}{T-NC}
\newprob{\AZC}{AllZeroColumn}

\renewcommand{\H}{\mathbf{H}}
\renewcommand{\Pr}{\mathbf{Pr}}
\newcommand{\Ex}{\mathbf{E}}
\newcommand{\Hmin}{\H_\infty}

\newcommand{\eqdef}{\coloneqq}

\newcommand{\AlgYZ}{\mathcal{A}^{\mathrm{YZ}}}
\newcommand{\AlgT}{\mathcal{A}^{\mathrm{T}}}
\newcommand{\Uadd}{U_{\mathrm{add}}}
\newcommand{\Udecode}{U_{\mathrm{decode}}}

\newcommand{\cX}{\mathcal{X}}

\renewcommand{\FF}{\mathbb{F}}
\renewcommand{\RS}{\mathrm{RS}}
\renewcommand{\GRS}{\mathrm{GRS}}
\newcommand{\fold}[1]{{#1}^{\mathsf{f}}}
\newcommand{\unfold}[1]{{#1}^{\mathsf{u}}}

\let\OLDthebibliography\thebibliography
\renewcommand\thebibliography[1]{
  \OLDthebibliography{#1}
  \setlength{\itemsep}{.4ex}
}

\newcommand{\Xcomment}[1]{{}}

\newcommand{\E}{\mathbb{E}}
\newcommand{\Bool}{\{0,1\}}

\newcommand{\bias}[1]{#1}
\newcommand{\unf}[1]{\overline{#1}}
\newcommand{\suc}{\mathsf{suc}}
\newcommand{\CN}{\mathtt{N}}
\newcommand{\polylog}{\mathrm{polylog}}

\begin{document}

\mbox{}\vspace{12mm}

\begin{center}
{\huge Quantum Communication Advantage in \TFNP}
\\[1cm] \large
	
\setlength\tabcolsep{1.4em}
\begin{tabular}{cccc}
Mika Göös&
Tom Gur&
Siddhartha Jain&
Jiawei Li\\[-1mm]
\small\slshape EPFL&
\small\slshape Cambridge &
\small\slshape UT Austin &
\small\slshape UT Austin

\end{tabular}
	
\vspace{2em}
	
\large
\today
	
\vspace{1em}
\end{center}
	
\begin{abstract}
\noindent
We exhibit a total search problem with classically verifiable solutions whose communication complexity in the quantum SMP model is exponentially smaller than in the classical two-way randomized model. Our problem is a bipartite version of a query complexity problem recently introduced by Yamakawa and Zhandry~({\footnotesize JACM 2024}). We prove the classical lower bound using the structure-vs-randomness paradigm for analyzing communication protocols.
\end{abstract}

\section{Introduction}

A foremost goal in quantum complexity is to demonstrate an exponential quantum advantage. A prevalent theme is that we know how to demonstrate exponential quantum advantage for ``structured'' problems. For unstructured search, Grover's algorithm famously only achieves a quadratic speedup, which we know to be tight~\cite{BBBV}. The canonical model for unstructured problems is \emph{total boolean functions}, $f\colon \Bool^n \rightarrow \Bool$, and we know that for any total boolean function the quantum queries needed to evaluate it is at most a 4th power of the number of classical queries \cite{ABKST21}.

The world of communication complexity is not immune to this difficulty, and a superpolynomial quantum communication speedup for total boolean functions is not known. The holy grail is again to have a speedup for a total boolean function, $f\colon \Bool^n \times \Bool^n \rightarrow \Bool$. Alas, we do not have any such candidates for a super-cubic speedup, and at the same time we cannot prove that quantum and classical communication complexities are polynomially related like in the query setting. 

The picture is completely different for partial functions and search problems, which have a promise about the structure of the problem input. A long line of work on partial problems~\cite{BCW98,Raz99,Yossef04, GKK+08,KR11,Gav16,Gav19,Gavinsky21,GirishRT22, li2024classical} has culminated in strong separations results: Gavinsky~\cite{Gav16} and Girish, Raz, and Tal~\cite{GirishRT22} exhibited partial functions with an exponential separation between the quantum SMP (simultaneous message passing) model and the classical two-way randomized communication model.

We employ a different way to impose structure on the problem: we consider total \NP search problems, which are relations lying in (the communication analogue of) the complexity class \TFNP. The seminal work of Bar-Yossef, Jayram, and Kerenidis~\cite{Yossef04} showed an exponential quantum--classical separation in the restricted setting of one-way protocols for a total \NP search problem (the \emph{hidden matching} problem). Since then, no progress has been made for $\TFNP$ problems.

In a recent breakthrough, Yamakawa and Zhandry~\cite{Yamakawa2022} introduced a new total \NP search problem that exhibits an exponential quantum speedup in the query complexity model. Their problem is also considered ``unstructured,'' because the quantum advantage holds over the \emph{uniform} distribution, a question raised by Aaronson and Ambainis~\cite{AA14}. In this work, we define a communication variant of their problem, called \emph{Bipartite NullCodeword}, and obtain the strongest quantum--classical communication separation so far for any total~\NP search problem.

\begin{restatable}[Main result]{theorem}{ThmMain}
\label{thm:main}
There exists a total two-party search problem $S\subseteq \{0,1\}^N\times\{0,1\}^N\times Q$ for $N = 2^{\Theta(n)}$ that
\begin{enumerate}
    \item admits a $\operatorname{poly}(n)$-bit protocol in the quantum SMP model, but
    \item requires $2^{n^{\Omega(1)}}$ bits of communication in the randomized two-way communication model (even under the uniform input distribution), and
    \item the relation is classically verifiable (in \NP) and thus the quantum advantage holds in \TFNP.
\end{enumerate}
\end{restatable}

The separation in \cref{thm:main} is qualitatively as strong as one would want. Recall that in a quantum SMP protocol, Alice and Bob send a quantum state to a third-party referee, Charlie, who outputs the answer without further interaction. Even the one-way model can simulate this since Bob can emulate Charlie. Furthermore, our quantum SMP protocol is \emph{computationally efficient}, i.e, all parties can be implemented using $\polylog (N)$ size quantum circuits with oracle access to the input. This is the same notion of efficiency as in the result of~\cite{GirishRT22}. Our protocol also does not use any shared entanglement or randomness, matching the result of Gavinsky~\cite{Gavinsky21}.

Besides being fundamental considerations in theoretical computer science, verifiability and efficiency are also desirable properties for a quantum--classical communication complexity separation to be used as an experiment to demonstrate \emph{unconditional} quantum advantage (a.k.a.\ quantum information supremacy \cite{ABK24}). We see our result as a theoretical feasibility result for designing such an experiment based on communication complexity. We compare our result to several other notable quantum--classical communication separations in \cref{tab:separations}.

\subsection{Other related work}
\paragraph{The Yamakawa--Zhandry problem.}

The Yamakawa--Zhandry problem~\cite{Yamakawa2022} was a breakthrough average-case query separation. Since then, it has proven to be useful in other settings, often by modifying the the problem slightly (like we do). Before our work, it was tweaked to make progress on the question of getting an oracle separation of \textsf{QMA} and \textsf{QCMA} \cite{li2024classical,BK24}. More relevant to our work, Li et al.~\cite{li2024classical} also use the Yamakawa--Zhandry problem to get a new separation between one-way quantum and one-way randomized communication complexity. However, their problem is easy in the classical two-way model since Bob's input is short.

\paragraph{Query-to-communication lifting.}

Most techniques in communication complexity literature which can distinguish quantum and classical communication complexity are tailored to promise problems. One technique that works for total problems (functions and relations) is query-to-communication lifting~\cite{GPW20,CFKMP21}.\footnote{Lifting theorems show that $\operatorname{CC}(f\circ g) = \Omega(\operatorname{Q}(f))$ where $g$ is a boolean function gadget, $\operatorname{CC}$ is a communication complexity measure and $\operatorname{Q}$ is a query complexity measure.} The current strongest separation for a \emph{total boolean function} between the quantum and classical two-way communication is cubic: This is achieved by lifting an analogous cubic separation in query complexity obtained by Sherstov, Storozhenko, and Wu~\cite[Theorem~1.5]{SSW23} and independently by Bansal and Sinha~\cite[Corollary 1.5]{BS21}.

The exponential quantum query speedup from the Yamakawa--Zhandry problem~\cite{Yamakawa2022} can also be lifted to communication (Lifted \NC in \cref{tab:separations}). Given this, we can obtain an exponential separation between quantum and classical two-way communication for an unstructured \emph{relation}. However, we do not know how to simulate queries to the input of the original relation for the lifted relation using less than two rounds of communication. This limits us to a two-way vs.\ two-way separation, as opposed to a quantum SMP vs.\ classical two-way separation. 

\paragraph{\NP relations vs.\ partial functions.}
A partial boolean function can always be converted to a total relation by allowing any output on inputs where the function is undefined; however, the resulting total relation has no guarantees about $\NP$ verifiability (even if we started with a verifiable partial function).\footnote{We thank an anonymous reviewer for this observation.}
Our separation for a total relation and the separation of~\cite{GirishRT22} for a partial function are formally incomparable because in our quantum upper bound, we require no entanglement, and our relation is additionally a total \NP relation. In a different context, a separation for a total \NP relation can be automatically translated to a separation for a partial function. Namely, Ben-David and Kundu~\cite{BenDavid24} show how to do this in the query complexity model by using a pointer-based construction. However, if we tried to implement a similar trick naively for our problem, we would lose the quantum easiness in the SMP model. 

\begin{table}[tb]
\centering

\newcommand{\strong}{\color{ForestGreen}}
\newcommand{\weak}{\color{BrickRed}}
\renewcommand{\arraystretch}{1.3}
\setlength\tabcolsep{0.29em}
\begin{tabular}{l l c c c c c}
\toprule[.5mm]
\bf Candidate problem & \bf Reference & \bf Quantum u.b. & \bf Classical l.b. &  \bf f / R & \bf Totality & \bf \NP \\
\midrule

Vector in Subspace&\cite{Raz99,KR11} & \weak{\bf one-way} &\strong{\bf two-way} &  \strong{\bf function} & \weak{\bf partial} & \weak \bf n\\
Gap Hamming Relation&\cite{Gavinsky21} & \strong \bf SMP & \strong \bf two-way & \weak \bf relation & \weak \bf partial  & \weak \bf n
\\
$\textsc{Forrelation}\circ \textsc{Xor}$ & \cite{GirishRT22} & \strong \bf SMP & \strong \bf two-way & \strong \bf function & \weak \bf partial & \weak \bf n
\\ 
Hidden Matching&\cite{Yossef04} & \weak{\bf one-way} & \weak{\bf one-way} & \weak\bf  relation & \strong \bf total & \strong \bf y\\
Lifted \NC&\cite{Yamakawa2022,GPW20} & \weak\bf  two-way & \strong \bf two-way &  \weak \bf relation & \strong \bf total & \strong \bf y
\\
\midrule
Bipartite \NC~~ & This work & \strong\bf  SMP & \strong \bf two-way & \weak \bf relation & \strong{\bf total} & \strong \bf y
\\
\bottomrule[.5mm]
\end{tabular}
\caption{Several notable exponential quantum--classical separations. {\strong \bf Green text} indicates a strong result and {\weak\bf red text} indicates a weak result.}
\label{tab:separations}
\end{table}

\section{Technical overview}
\label{sec:overview}

We construct a relation to prove \cref{thm:main} in two stages. First, we define a natural bipartite analogue of the relation studied by Yamakawa--Zhandry which we call \emph{Bipartite NullCodeword} (\bNC). We show in \cref{sec: binc,sec:lower-bound} that for \bNC there is an average case separation between QSMP and two-way randomized communication for a certain code and distribution. We then convert this into a worst-case separation for a total relation in \cref{sec:total}.

\subsection{The Yamakawa--Zhandry problem and its bipartite variant}

We build upon the Yamakawa--Zhandry query complexity problem~\cite{Yamakawa2022}, which we call the \NC (\NCshort) problem.
In this problem, an error correcting code $C \subseteq \Sigma^n$ is chosen, where~$\Sigma$ is an exponential-size alphabet, namely $|\Sigma|=2^{\Theta(n)}$. The inputs are $n$ random oracles~$H_1, \ldots, H_n \colon \Sigma \to \Bool$, and the goal is to find a codeword $x \in C$ such that $H(x) = 0^n$, where~$H\colon \Sigma^n \to \Bool^n$ is the concatenation of all $H_i$s. 

\begin{lemma}[Informal, \cite{Yamakawa2022}]
    When $C$ is instantiated by a {folded Reed--Solomon} code with certain parameters, there exists a $\poly(n)$-query quantum algorithm (called the Yamakawa--Zhandry  algorithm) solving \NC, but any classical algorithm requires $2^{n^{\Omega(1)}}$ queries. 
\end{lemma}

Our communication problem is a bipartite analogue of \NC which we call \bNC. In \bNC, we give the first half of the oracles $H_1, \ldots, H_{n/2}$ as input to Alice, and the second half $H_{n/2+1}, \ldots, H_n$ as input to Bob. Their goal is still to find a codeword $x \in C$ such that $H(x) = 0^n$. The input length for Alice and Bob is $N \coloneqq n|\Sigma|/2 = 2^{\Theta(n)}$.

\paragraph{Quantum easiness of \bNC.}
One immediate advantage of our definition of \bNC is that it is fairly easy in the quantum communication setting. The key observation is that the Yamakawa--Zhandry  algorithm for \NC can be divided into two stages---it first makes \emph{non-adaptive queries} to each $H_1, \ldots, H_n$ individually, and then \emph{processes} the results without making any more queries. Therefore, we can get a QSMP protocol of \bNC by simulating the Yamakawa--Zhandry algorithm: Alice and Bob first simulate the non-adaptive queries to their halves of inputs and send their results to Charlie; Charlie then simulates the processing stage of the Yamakawa--Zhandry algorithm and outputs the solution. Furthermore, this QSMP protocol is computationally efficient and uses no shared entanglement/randomness.

\subsection{Classical lower bound for \bNC}

Proving a classical communication lower bound for \bNC turns out to be a more challenging task. 
First, most of the communication lower bound techniques we currently have only work for boolean functions (or relations with few or unique solutions), but \bNC has exponentially many solutions for a typical random input. Furthermore, the choice of the error correcting code $C$---which is a folded Reed--Solomon code---complicates the analysis for proving a classical lower bound. For example, information complexity seems like a viable approach if $C$ is a random code. But the analysis would become intimidating if a folded Reed--Solomon code is considered. 

Below, we give an overview of our proof. For $X\subseteq \Bool^N$ and $I\subseteq [N]$, we use $X_I\coloneqq \{ x_I\in\Bool^I : x\in X\}$ to denote the set of strings in $X$ restricted to the coordinates in $I$.

\paragraph{Simple special case: Subcube protocols.} We start by proving a lower bound for highly structured protocols. A rectangle $X\times Y \subseteq \Bool^N \times \Bool^N$ is said to be a subcube rectangle if there exist two index sets $I, J \subseteq [N]$ such that $X_I$ and $Y_J$ are \emph{fixed} strings and the sets $X_{\bar{I}}=\Bool^{\bar{I}}$ and~$Y_{\bar{J}}=\Bool^{\bar{J}}$ contain all possible strings. A subcube protocol is such that every node of the protocol tree corresponds to a subcube rectangle. This model is closely related to decision trees (but strictly more general). The below property of the code $C$ will be central to the lower bound.

\begin{definition}[List-recoverability: Simplified definition] \label{def:list-rec-simple}
    A code $C\subseteq \Sigma^n$ is \emph{list-recoverable} if for any collection of subsets $S_1, \ldots, S_n \subseteq \Sigma$ such that $\sum_i |S_i| \leq \ell$, it holds that 
\begin{equation}\label{eq:strong-lr}
    |\{(x_1, \ldots, x_n) \in C: |i \in [n]: x_i \in S_i| \geq 0.4 n\}| \leq 2^{o(n)} \:.
\end{equation}
\end{definition}

Suppose $\Pi$ is a subcube protocol that communicates at most $\ell \leq 2^{o(n)}$ bits. Assume also for simplicity that all rectangles in the subcube protocol fix at most $\ell$ coordinates.  It follows from~\eqref{eq:strong-lr} that for any rectangle $R$ in the protocol tree of $\Pi$, the number of codewords $x\in C$ such that more than $0.4n$ bits of $H(x)$ are fixed in $R$ is at most $2^{o(n)}$.

Let us now sketch how to use this property to bound the success probability of $\Pi$. Say that a codeword $x\in C$ is \emph{dangerous} for a rectangle $R$ if at least $0.4n$ of the bits of $x$ are fixed by $R$; intuitively, this means that the protocol has communicated a lot of information about $x$. Consider any leaf rectangle $R$ and suppose it is labeled with solution $x\in C$. Assuming for contradiction that~$\Pi$ never errs, then $x$ is a correct solution (all-0 codeword) for all inputs $H\in R$; in particular, $x$ is dangerous for $R$.  We have two key facts:
\begin{itemize}
    \item \emph{Few dangerous codewords:} By list-recoverability~\eqref{eq:strong-lr}, the number of codewords $x$ that ever become dangerous in a single execution of $\Pi$ is at most $2^{o(n)}$.
    \item \emph{Dangerous codewords rarely become solutions:} Consider a single execution of the protocol and suppose Alice sends a message that causes a given codeword $x$ to become dangerous for the first time at rectangle $R$. Then there are at least $0.1n$ unfixed input bits of~$x$ on Bob's side. Thus, at most $2^{-0.1n}$ fraction of inputs in $R$ would map $x$ to $0^n$.
\end{itemize}
We can conclude our lower bound for subcube protocols by a union-bound argument: None of the dangerous codewords will likely be a valid solution at the leaf and hence $\Pi$ errs often.

\paragraph{Structure-vs-randomness dichotomy.}

To prove a lower bound for general protocols, we use the structure-vs-randomness framework developed by Wang, Yang, and Zhang~\cite{YZ23, WYZ23} who build upon query-to-communication lifting techniques~\cite{GPW20}. We use this framework to convert any general communication protocol into a \emph{subcube-like protocol}. The nodes of a subcube-like protocol correspond to \emph{subcube-like rectangles} $X \times Y \subseteq \Bool^N \times \Bool^N$ defined such that
\begin{enumerate}
    \item $X_I, Y_J$ are \emph{fixed} strings for some index sets $I,J\subseteq[N]$. \hfill (\emph{structure})
    \item $X_{\bar{I}}, Y_{\bar{J}}$ are pseudorandom (they ``look'' like the uniform distribution). \hfill (\emph{randomness})
\end{enumerate}
We show that subcube-like protocols behave similarly enough to subcube protocols that we may simply re-do the proof sketched above for subcube protocols.

\paragraph{Biased input distribution.}

So far, we have been ignoring a non-trivial issue---the list-recoverability property in \cref{def:list-rec-simple} required by our argument is not known to hold for the folded Reed--Solomon code $C$ used for the original Yamakawa--Zhandry problem. To improve the parameters in the list-recoverability of the folded Reed--Solomon code, one could decrease its degree parameter $k$. However, a smaller degree $k$ would also weaken the list-decodability of its dual code, which is essential for the quantum upper bound of the \NC problem. This inherent tension between the list-recoverability and the dual-decodability poses the following dilemma: if we decrease the degree $k$ to an extent that list-recoverability is strong enough for our lower bound argument, the dual-decodability would be too weak to retain our quantum upper bound.

To resolve this, we tweak the input distribution of \NC and \bNC. Recall that the \NC problem in \cite{Yamakawa2022} is defined over uniform random oracles. We now consider the scenario where the input bits can be \emph{biased}, i.e., each symbol is (independently) mapped to $1$ with probability $p$ for a constant $p \in (0,1)$. When $p$ gets smaller, it becomes easier to find an all-0 codeword. By generalising the original analysis in \cite{Yamakawa2022}, we show that a weaker dual-decodability suffices for the quantum upper bound when $p$ is a small constant. Therefore, we can decrease the degree $k$ of the code to achieve the list-recoverability parameters we need for our lower bound analysis while maintaining the quantum upper bound under the biased input distribution.

It might be of independent interest that our analysis broadens the range of parameters for which there is an exponential quantum advantage in the Yamakawa--Zhandry~\cite{Yamakawa2022} problem.

\paragraph{Conversion to a total problem.} Our discussion so far has been about an average case separation of quantum and classical communication complexity. Our input distribution has a non-zero probability of having \emph{no solution}. Thus, we do not have totality yet. To fix this, we borrow a trick from the original \cite{Yamakawa2022} paper. The idea begins with the observation that an $O(n)$-wise independent distribution suffices for the quantum upper bound. Thus, instead of a uniform random input we could give arbitrary inputs and allow Alice and Bob to $\textup{XOR}$ their inputs with a string that they can choose from an $O(n)$-wise independent hash family. The canonical family is that of low-degree polynomials over finite fields~\cite{Vadhan12psr}, which can be efficiently implemented.

However, modifying the problem in this way also allows for a \emph{classical} upper bound! The protocol can simply pick an arbitrary codeword $x\in C$ and run Gaussian elimination to find a function $h$ from the family which maps the codeword $x$ to $0^n$. To get around this, we further modify the problem and require the protocol to solve $t$ instances of $\bNC$ that are $\textup{XOR}$ed with the \emph{same} hash function. 

To prove this lower bound, we first apply a direct product theorem for product distributions~\cite{braverman2013} and show that for each fixed choice of hash function, the success probability of any sub-exponential cost classical protocol is less than $\approx 2^{-t}$. 
Thus, by taking $t$ as a large polynomial in $n$, we can use a \emph{union bound} over all hash function from the family and show that the success probability is still $o(1)$ even when the protocol is given the freedom of picking a hash function.

\section{Preliminaries}

In this work, we study \emph{search} problems, formalized by relations. We introduce terminology below.

A search problem $R$ is a relation $R \subseteq [k]^N \times O$ where for any given $x \in [k]^N$ we want to find $y \in O$ such that $(x,y) \in R$. A two-party search problem $R$ is a bipartite relation $R \subseteq [k]^N \times [k]^N \times O$ where given $x, y \in [k]^N$ the two parties want to find a $z$ such that $(x,y,z) \in R$. When the two-party search problem is total and efficiently verifiable, it is said to be in communication-\TFNP (see \cite{GKRS19,Rezende2022} for a discussion of this class).

\subsection{Communication complexity}
We assume some familiarity with quantum computing and communication complexity. For more background refer to textbooks~\cite{NCbook,Kushilevitz1997,Rao2020}. We define QSMP and two-way randomized communication protocols below.

\begin{definition}[QSMP]
    A QSMP (without shared randomness/entanglement) protocol $\Pi$ for relation $R \subseteq X\times Y\times O$ is pair of quantum states $\ket{\phi_x}, \ket{\phi_y}$ for every input $(x,y) \in X\times Y$. Alice and Bob send $\ket{\phi_x}$ and $\ket{\phi_y}$ respectively to Charlie, and he outputs a solution $\Pi(x, y) = \bm{o}_{x,y}$. The protocol is said to succeed with probability $1 - \epsilon$ if the probability that $(x,y,\bm{o}_{x,y}) \not\in R$ is at most $\epsilon$. The communication cost (denoted $|\Pi|$) is the total number of qubits sent by Alice and Bob on the worst-case input.
\end{definition}

\begin{definition}[Two-way randomized]
    A (public-coin) two-way randomized protocol $\Pi$ is a sequence of messages exchanged between Alice and Bob which depend on their input, previous messages, and publicly sampled randomness after which Bob outputs $\Pi(x,y) = \bm{o}_{x,y}$. The protocol is said to succeed with probability $1 - \epsilon$ if the probability that $(x,y,\bm{o}_{x,y}) \not\in R$ is at most $\epsilon$. The communication cost (denoted $|\Pi|$) is the total number of bits sent by Alice and Bob on the worst-case input.
\end{definition}

\subsection{Error-correcting codes}

    For a prime power $q$, 
    we denote by $\FF_q$ the finite field of order $q$ and denote by $\FF_q^{\leq k}[x]$ 
    the set of polynomials over $\FF_q$ with degree at most $k$. For any vector $x$, define $\hw(x)$ as the Hamming weight of $x$, i.e., the number of non-zero elements in $x$.

    An error correcting code $C$ of length $\CN$ and distance $d$ over the alphabet $\Sigma$ is a subset of $\Sigma^{\CN}$ such that the Hamming distance of every $c_1,c_2 \in C$ is lower bounded by $d$. We will use the notation $\CN$ or $n$ for the code length in order to distinguish with the input length $N$ of the communication problem.
    
    We will focus on linear codes, which are codes that are linear subspaces.

    \begin{definition}[Linear code and its dual code]
        A code $C \subseteq \FF_q^{\CN}$ is a linear code if $C$ is a linear subspace of $\FF_q^{\CN}$. The dual code $C^{\bot}$ of a linear code $C$ is defined as its orthogonal complement:
        \[ C^{\bot} = \{x \in \FF_q^{\CN} : x \cdot y = 0, \forall y \in C \} . \]
    \end{definition}
    
    \begin{definition}[Generalized Reed--Solomon code]
        The Reed--Solomon code $\RS_{\FF_q, \gamma, k}$ over $\FF_q$ with generator $\gamma \in \FF_q^*$ and degree $0 \leq k \leq \CN$ has length $\CN \coloneqq q-1$, and is defined as:
        \[ \RS_{\FF_q, \gamma, k} \coloneqq \{(f(0), f(\gamma), \ldots, f(\gamma^{\CN-1})): f \in \FF_q^{\leq k}[x] \}. \]

        For any vector $v = (v_0, \ldots, v_{\CN-1}) \in {\FF_q^*}^{\CN}$, the generalized Reed--Solomon code $\GRS_{\FF_q, \gamma, k, v}$ is defined as:
        \[ \GRS_{\FF_q, \gamma, k, v} \coloneqq \{(v_0f(0), v_1f(\gamma), \ldots, v_{\CN-1}f(\gamma^{\CN-1})): f \in \FF_q^{\leq k}[x] \}. \]
    \end{definition}

    The Reed--Solomon code is a special case of the generalized Reed--Solomon code with $v = (1, \ldots, 1)$. The generalized Reed--Solomon code is a linear code, and the dual code of $\RS_{\FF_q, \gamma, k}$ is $\GRS_{\FF_q, \gamma, \CN-k-2, v}$ for some $v \in {\FF_q^*}^{\CN}$.

    Since multiplying a scalar on each coordinate does not affect the decodability, the classical list-decoding algorithm for the Reed--Solomon code also applies to the generalized Reed--Solomon code.

    \begin{lemma}[List-decodability of the GRS code~\cite{GS99}]\label{lem: GRS list-decode}
        There is a deterministic list-decoding algorithm $\listdecode$ for $\GRS_{\FF_q, \gamma, k,v}$ such that for any $z \in \FF_q^{\CN}$, $\listdecode(z)$ returns a list of all codewords with hamming distance at most $\CN - \sqrt{k\CN}$ to $z$ in $\poly(\CN)$ time.
    \end{lemma}

    \begin{definition}[Folded linear codes]
        For a linear code $C \subseteq \FF_q^{\CN}$ and an integer $m$ that divides $\CN$, the $m$-folded code $C^{(m)}$ is a linear code over alphabet $\Sigma = \FF_q^m$ of length $n \coloneqq \CN / m$, defined as follow.
        \[ C^{(m)} = \{( (x_1, \ldots, x_m), (x_{m+1}, \ldots, x_{2m}), \ldots, (x_{\CN-m+1}, \ldots, x_\CN)): (x_1, \ldots, x_\CN) \in C \}. \]
    \end{definition}

    Note that $(C^{\bot})^{(m)} = (C^{(m)})^{\bot}$. In particular, the dual code of folded Reed--Solomon code $\RS_{\FF_q,\gamma,k}^{\mfol}$ is the folded generalized Reed--Solomon code $\GRS_{\FF_q,\gamma,d, \CN - k - 2, v}^{\mfol}$ for some $v \in \FF_q^{\CN}$.
    
    \begin{definition}[List-recoverability]
        A code $C \subseteq \Sigma^n$ is $(\zeta, \ell, L)$ list-recoverable if for any collection of subsets $S_1, \ldots, S_n \subseteq\Sigma$ such that 
        $|S_i| \leq \ell$ for any $i \in [n]$, it holds that        \[ |\{(x_1, \ldots, x_n) \in C: |i \in [n]: x_i \in S_i| \geq \zeta n\}| \leq L \:.\]
    \end{definition}

    The folded Reed--Solomon code is known to be list-recoverable in certain parameter regimes.

    \begin{lemma}[\cite{GR08, rudra2007list}]\label{lem: list recover of FRS}
        Let $\CN = q-1$. For positive integers $\ell, r, s, m$, where $s \leq m$, and a real number $\zeta \in (0,1)$, the code $\RS^{(m)}_{\FF_q, \gamma, k}$ is $(\zeta, \ell, q^s)$ list-recoverable if the following inequalities hold:
        \begin{equation}\label{equ: list recover 1}
            \frac{\zeta \CN}{m} \geq (1+\frac{s}{r}) \cdot \frac{(\CN\ell k^s)^{\frac{1}{s+1}}}{m-s+1}.
        \end{equation}
        \begin{equation}\label{equ: list recover 2}
            (r+s)\cdot (\CN\ell / k)^{\frac{1}{s+1}} < q.
        \end{equation}
    \end{lemma}

    \paragraph{Quantum fourier transform over the code space.}

    The quantum Fourier transform over the finite field $\FF_q$ where $q = p^r$ for a prime $p$, is a unitary defined as follows:
    \[\QFT_{\FF_q} \ket{x} = \frac{1}{\sqrt{q}} \sum_{z \in \FF_q} \omega_p^{\Tr(x \cdot z)} \ket{z}, \]
    where the trace function is defined as $\Tr(x) \coloneqq \sum_{i=0}^{r-1} x^{p^i}$ for any $x \in \FF_q$.

    The quantum Fourier transform over the alphabet $\Sigma = \FF_q^m$ is defined as the $m$-tensor product of $\QFT_{\FF_q}$:

    \[ \QFT_{\Sigma} \ket{x} \coloneqq \QFT_{\FF_q}^{\otimes m} \ket{x_1}\ket{x_2} \ldots \ket{x_m} =  \frac{1}{\sqrt{|\Sigma|}} \sum_{z \in \Sigma} \omega_p^{\Tr(x \cdot z)} \ket{z}.\]

    Note that the second equality uses the linearity of $\Tr(\cdot)$. Similarly, the quantum Fourier transform over the code space $\Sigma^n$ is

    \[\QFT\ket{x} \coloneqq \QFT_{\Sigma}^{\otimes n}\ket{x} = \frac{1}{|\Sigma|^{n/2}} \sum_{z \in \Sigma^{n}} \omega_p^{\Tr(x \cdot z)} \ket{z}.\]

    For a quantum state $\ket{\phi} = \sum_{x} f(x) \ket{x}$, we use $\hat{f}$ to denote its fourier coefficient, i.e., 
    \[\QFT \ket{\phi} = \sum_{z} \hat{f}(z) \ket{z}.\] 

\subsection{The Yamakawa--Zhandry problem}
The Yamakawa--Zhandry problem is a search problem in which, given a fixed code and query access to random oracles, the goal is to find a codeword that is null for the oracles.

\begin{definition}[\NC~\cite{Yamakawa2022}]\label{def:YZ}
    Fix a family of codes $C = \{C_n\}$, where $C_n \subseteq \Sigma^n$. For each $i \in [n]$, let $H_i\colon \Sigma \to \Bool$, and let $H\colon \Sigma^n \to \Bool^n$ be the concatenation $H(x) = (H_1(x_1), \ldots, H_n(x_n))$.
    The relation $\YZ^C_n \subseteq \Bool^{n|\Sigma|} \times \Sigma^n$ is defined as:
    \[ \YZ^C_n = \{(H,x) : x \in C \land H(x) = 0^n  \}. \]
    
    The \NC problem $\YZ^C$ is defined as the following search problem: Given $n$ and the oracle access to $H \in \Bool^{n|\Sigma|}$, find a string $x$ such that $(H, x) \in \YZ^C_n$.
\end{definition}

    A folded Reed--Solomon code $C$ with certain parameters is chosen in \cite{Yamakawa2022}, which allows $\YZ^C$ to have a quantum algorithm with $\poly(n)$ queries and at the same time be exponentially hard for any classical (randomized) algorithm, when the input $H$ is given by a uniformly random oracle.

\section{The Bipartite NullCodeword problem}\label{sec: binc}
We present the formal definition of the Bipartite NullCodeword communication problem (\bNC) in \autoref{sec: binc: def}. We then, in \autoref{sec: binc: code}, instantiate the problem with the folded Reed-Solomon code, with a careful parameterization that provides the properties needed for our techniques. Finally, we provide an efficient quantum SMP protocol for \bNC in \autoref{sec: binc: quantum easiness}. Throughout this section, we will work with a biased-input distribution. In \autoref{sec:total}, we will convert this into a total problem.

\subsection{Definition of \bNC} \label{sec: binc: def}

To transform the \NC problem (\autoref{def:YZ}) into a communication problem, the mapping $H\colon \Sigma^{n} \to \Bool^{n}$ is divided into two halves $H_A, H_B: \Sigma^{n/2} \to \Bool^{n/2}$ as the input of Alice and Bob respectively.

\begin{definition}[Bipartite NullCodeword, \bNC]
\label{defn:binc}
    Fix a family of codes $C = \{C_n\}$, where $C_n \subseteq \Sigma^n$. Let $H_1 \ldots, H_n$ be $n$ mappings from $\Sigma$ to $\Bool$,
    define $H_A, H_B\colon \Sigma^{n/2} \to \Bool^{n/2}$, where $H_A = (H_1,\ldots,H_{n/2})$ and $H_B = (H_{n/2 +1},\ldots,H_{n})$.
    The bipartite relation
    $\bNC^C_n \subseteq\Bool^{n|\Sigma|/2} \times \Bool^{n|\Sigma|/2} \times \Sigma^n$ is defined as:
    \[ \bNC^C_n = \{(H_A, H_B, x) : x \in C_n \:\land\: H_A(x_A) = H_B(x_B) = 0^{n/2}  \}\:, \]
    where $x_A = (x_1,\ldots,x_{n/2})$ and $x_B = (x_{n/2+1},\ldots,x_{n})$.

\end{definition}

The superscript $C$ may be omitted in $\bNC^C$ if the choice of code $C$ is clear from the context.
We will consider a distributional version of $\bNC^C$ where the input is given by a \emph{biased random oracle}.

\begin{definition}[{$p$}-biased input distribution]
\label{defn: biased oracle}
    A distribution over $H = (H_1, H_2, \ldots, H_n)$ where each $H_i$ maps every element of $\Sigma$ to $1$ with probability $p$ independently.
\end{definition}

We will pick the constant $p \coloneqq 2^{-6}$ for our setting, though we note that for any constant value below 1/2, we can get a non-trivial tradeoff and preserve the classical lower bound.

\subsection{A suitable code: Trading off dual-decodability for list-recoverability}\label{sec: binc: code}

    We modify the code used in the \NC problem~\cite{Yamakawa2022} to achieve a stronger list-recoverability, which we will need to prove the classical two-way communication lower bound. We do this by reducing the \emph{degree} of the polynomials in the construction of the Reed--Solomon code. This hampers the decodability of the dual code, but we show that over the biased input distribution (\cref{defn: biased oracle}), we can still decode the dual code with high probability, which is necessary for our upper bound. We instantiate the code we will use for the remainder of this paper below.

    \begin{definition}\label{def: code parameters}
        Define code $C = \{C_n\}$ as the folded Reed--Solomon codes $\RS^{(m)}_{\FF_q, \gamma, k}$ 
        with the following choice of parameters: 
        \begin{itemize}
            \item the length of the code $n = 2^t - 1$ for some $t \in \mathbb{\CN}$; 
            \item $q = 2^{2t}$, and thus, $\CN = q-1$ and $m = \CN/n = 2^t+1$;
            \item $\gamma$ is an arbitrary generator of $\FF_q^*$;
            \item degree $k = 0.1\CN$.
        \end{itemize}
        
        By definition, $|C_n| = q^{k+1} =  2^{\Theta(n^2)}$ and $|\Sigma| = |\FF_q^m| = q^m = 2^{\Theta(n)}$.
    \end{definition}
    
    Only the degree parameter $k$ is different compared with the original \NC problem~\cite{Yamakawa2022}, in which $k$ is chosen to be $\alpha \CN$ for some constant $\alpha \in (5/6, 1)$. 
    Below, we state a lemma with the parameters for which our code is list-recoverable and dual-decodable over the $p$-biased input distribution.
    
    \begin{lemma}[Modified from Lemma 4.2 in \cite{Yamakawa2022}]
    \label{lem:code}
        The code $C$ 
        satisfies:
        \begin{enumerate}   
            \item $C_n$ is $(\zeta, \ell, L)$-list-recoverable, where $\zeta = 0.4, \ell = 2^{n^{0.8}}, L = 2^{n^{0.9}}$. \hfill (List-Recoverability)
            \item Let $\mathcal{D}_n$ be the product distribution on $\Sigma^n$ where each symbol is $0$ with probability $1 - p$ and otherwise a uniformly random element of $\Sigma \setminus \{0\}$. If $p \leq 0.04$, there exists a $\poly(n)$ time deterministic decoding algorithm $\decode$ for the dual code $C^{\bot}_n$ such that 
            \begin{equation}\label{equ: unique decode}
                \mathop{\Pr}_{e \sim \mathcal{D}_n}[\forall x \in C^{\bot}_n, \decode(x + e) = x] = 1 - 2^{-\Omega(n)}.
            \end{equation}
            \hfill (Unique-Decodability of the Dual)
        \end{enumerate}
    \end{lemma}

    \begin{proof}
    The proof of (1), the list-recoverability property, follows immediately from \autoref{lem: list recover of FRS}. Namely,
    we apply \autoref{lem: list recover of FRS} with $s \coloneqq n^{0.89}, r \coloneqq n^{0.95}$. It holds that $\frac{\zeta \CN}{m}  = 0.4n$, and
    \[(1+\frac{s}{r}) \cdot \frac{(\CN\ell k^s)^{\frac{1}{s+1}}}{m-s+1} = (1 + n^{-0.06})\cdot \frac{(10k\cdot 2^{n^{0.8}} \cdot k^{s})^{\frac{1}{s+1}}}{n-n^{0.89}+3} \approx \frac{k}{n} = 0.1n, \]
    which proves inequality~\eqref{equ: list recover 1}. And for inequality~\eqref{equ: list recover 2}, 
    \[ (r+s)\cdot (\CN\ell / k)^{\frac{1}{s+1}} = (n^{0.95} + n^{0.89}) \cdot (10 \cdot 2^{n^{0.8}})^{\frac{1}{s+1}} < n < q. \]
    Finally, note that $L = 2^{n^{0.9}} \geq q^s = 2^{n^{0.89}\log q}$.

    It remains to show the unique decodability of the dual code. 
    Note that the dual code $C^\bot_n$ is the folded generalized Reed--Solomon code $\GRS_{\FF_q,\gamma,d,v}^{\mfol}$ for some $v \in \FF_q^{\CN}$, where $d \coloneqq \CN - k - 2 = 0.9\CN - 2$. 
        Let $\varepsilon \coloneqq 0.01$.
    For any vector $x \in \FF_q^{\CN}$, let $\fold{x} \in \Sigma^n$ be the $m$-folded version of $x$; similarly, for any vector $z \in \Sigma^n$, define $\unfold{z} \in \FF_q^{\CN}$ by unfolding each letter in $z$. So, if $z = x_1, \ldots, x_m \in \FF^n_q$ then $\unfold{z} = (x_1, \ldots, x_m) \in \FF^{\CN}_q$.

    The algorithm $\decode(z)$ works as follows: 
    On input $z \in \Sigma^n$, it first unfolds $z$ into $\unfold{z} \in \FF_q^{\CN}$. It then runs the list-decoding algorithm $\listdecode(\unfold{z})$ (\autoref{lem: GRS list-decode}) for the generalized Reed--Solomon code $\GRS_{\FF_q,\gamma,d,v}$ to get a list of codewords. If there is a unique $x \in \FF_q^{\CN}$ in the list such that $\hw(\unfold{z}-x)\leq (p+\varepsilon)\CN$, it outputs $\fold{x}$ (interpreted as a codeword of $\GRS_{\FF_q,\gamma,\CN-k-2,\vv}^{\mfol}$), 
    and otherwise outputs $\bot$.   

    We now analyze the success probability of $\decode$. 
    Define a subset $\gooderrors\subseteq \Sigma^n$ of \emph{good} error terms, leading to unique decoding, as follows.
    \begin{equation*}
     \gooderrors \coloneqq \{e \in \Sigma^n: \hw(e)\leq (p+\varepsilon)n \}.  
    \end{equation*}
    Applying a Chernoff bound
    to the Hamming weight of $e$, we know $e$ is good with high probability. 
    \begin{equation}\label{equ: good errors}
    \Pr_{\ve \sim \dist_n}[e \in \gooderrors] \geq 1 - 2^{-\Omega(n)}.
    \end{equation}
    Therefore, it suffices to prove that the algorithm $\decode$ always succeeds on a good error. Towards this end, we first show a useful property of set $\gooderrors$. 
        \begin{claim}\label{lem: good error}
            For any $\ve\in \gooderrors$ and $y \in C^\perp_n \setminus \{ 0 \}$, we have 
            $\hwUF(\unfold{e}-\unfold{y})>(p+\varepsilon)\CN$.
        \end{claim}
        \begin{proof}
            For any $\ve\in \gooderrors$, $\hwUF(\unfold{e})\leq (p+\varepsilon)\CN < 0.05N$.
            Note that for any non-zero codeword $y \in C^\perp_n$, it holds that $\hwUF(\unfold{y}) \geq \CN - d \geq 0.1\CN$, because a degree $d$ non-zero polynomial could have at most $d$ roots. By the triangle inequality, we have $\hwUF(\unfold{e}-\unfold{y}) > 0.05\CN \geq (p+\varepsilon)\CN$.
        \end{proof}

    Using \autoref{lem: good error}, we proceed to show the correctness of the algorithm $\decode$.
    \begin{claim}\label{lem: unique decodability}
        For any $\vx\in C^\perp_n$ and $\ve\in \gooderrors$, $\decode(\vx+\ve)=\vx$.
    \end{claim} 
    
    \begin{proof}
        Fix any $\vx\in C^\perp_n$ and $\ve\in \gooderrors$.
        Recall that $d = 0.9\CN - 2$, it holds that $\CN-\sqrt{d\CN} > 0.05 \CN \geq (p+\varepsilon)\CN$. By \autoref{lem: GRS list-decode}, $\unfold{x}$ is contained in the list of codewords outputted by the list-decoding algorithm for $\GRS_{\FF_q,\gamma,d,v}$.
        For the uniqueness, consider any $y \in C^\perp_n, y \neq x$, since $C^\perp_n$ is a linear code, we have
        \[(x + e) - y = e - (y - x) = e - y'\] for some $y' \in C^\perp_n \setminus \{ 0 \}$.
        By \autoref{lem: good error}, we have that $\hwUF(\unfold{(x+e)}-\unfold{y}) = \hwUF(\unfold{e}-\unfold{(y')}) >(p+\varepsilon)\CN$. Thus, $\unfold{x}$ is the only codeword in the list whose Hamming distance from $\unfold{(x+e)}$ is smaller than or equal to $(p+\varepsilon)\CN$.
        We conclude that $\decode(\vx+\ve)=\vx$.
    \end{proof}
    
    The proof of the second part of \autoref{lem:code} then follows from \autoref{lem: unique decodability} and inequality~\eqref{equ: good errors}.
    \end{proof}

\paragraph{Remark.} In the following subsection and \cref{sec:total}, we will denote \bNC to be the problem with $C$ fixed to be the code in \cref{def: code parameters}, dropping the superscript. 

\subsection{Quantum SMP easiness of \bNC}\label{sec: binc: quantum easiness}
    We prove the quantum easiness of $\bNC$ by adapting the query algorithm for the \NC problem to the biased bipartite setting, noting that the protocol can be implemented under the constraints of the quantum SMP model.

    \begin{restatable}{theorem}{QuantumEasiness}\label{thm: quantum easiness}
        $\bNC$ has a quantum SMP protocol using $\poly(n)$ qubits over the $p$-biased distribution.

    \end{restatable}

\begin{algorithm}[ht]
\caption{Quantum SMP Protocol for $\bNC$}\label{alg:smp_protocol}
\begin{algorithmic}[0]
\State \textbf{Notation:} 
    \State For each $i \in [n]$, define the state
    \[
    \ket{\phi_i} \propto \sum_{e \in \Sigma: \; H_i(e) = 0} \ket{e}. 
    \]

\State Let $\Uadd, \Udecode$ be unitaries defined as follows:
\[
\ket{x}\ket{e} \xrightarrow{\Uadd} \ket{x}\ket{x+e} \xrightarrow{\Udecode} \ket{x - \decode(x+e)}\ket{x+e}.
\]

\Statex \noindent \rule{\linewidth}{0.4pt}

\State \textbf{Alice:} 
    \State Prepare $\ket{\phi_1}, \ldots, \ket{\phi_{n/2}}$ using $H_1, \ldots, H_{n/2}$.
    \State Send $\ket{\phi_1}, \ldots, \ket{\phi_{n/2}}$ to Charlie.

\Statex

\State \textbf{Bob:}
    \State Prepare $\ket{\phi_{n/2+1}}, \ldots, \ket{\phi_n}$ using $H_{n/2+1}, \ldots, H_n$.
    \State Send $\ket{\phi_{n/2+1}}, \ldots, \ket{\phi_n}$ to Charlie.

\Statex

\State \textbf{Charlie:}
    \State After receiving $\ket{\phi} \coloneqq \ket{\phi_1} \otimes \ldots \otimes \ket{\phi_n}$, perform the following steps:
    \begin{enumerate}
        \item Prepare a state:
        \[
        \ket{\psi} \propto \sum_{x \in C} \ket{x}.
        \]
        \item Apply $\QFT$ to both $\ket{\psi}$ and $\ket{\phi}$. Now, Charlie holds the state:
        \[
        \ket{\eta} \coloneqq \QFT\ket{\psi} \otimes \QFT\ket{\phi}.
        \]
        \item Apply $(I \otimes \QFT^{-1}) \Udecode \Uadd$ to $\ket{\eta}$.
        \item Measure the second register and output the measurement result.
    \end{enumerate}
\end{algorithmic}
\end{algorithm}

    We provide a proof sketch here and refer the reader to \autoref{sec: full proof of quantum easiness} for the full proof of the correctness of our quantum SMP protocol, as captured in \autoref{thm: quantum easiness}.

    \begin{proof}[Proof Sketch]
        As pointed out by Liu~\cite{Liu23}, the original \NC problem can be solved by a two-stage quantum algorithm $\AlgYZ$:
        \begin{description}
            \item[Query] Make $\poly(n)$ queries to each of $H_1, \ldots, H_n$, and store the raw results from those queries into states $\ket{\phi_1}, \ldots, \ket{\phi_n}$, where $\ket{\phi_i}$ stores the query results from $H_i$.
            \item[Process] Process $\ket{\phi_1}, \ldots, \ket{\phi_n}$ and output the solution without making any more queries to $H$.
        \end{description}

        Our quantum SMP protocol for $\bNC^C$ emulates the query protocol as follows. Alice and Bob first prepare states $\ket{\phi_1}, \ldots, \ket{\phi_{n/2}}$ and $\ket{\phi_{n/2+1}}, \ldots, \ket{\phi_n}$ by simulating the \textbf{Query} stage of $\AlgYZ$ on their half of inputs, and then send them to Charlie. Then, Charlie runs the same \textbf{Process} stage of $\AlgYZ$ with states $\ket{\phi_1}, \ldots, \ket{\phi_n}$ and outputs the solution. See \cref{alg:smp_protocol} for the quantum SMP protocol in detail.
        
        It remains to argue that the algorithm $\AlgYZ$ still works for $\YZ^C$ after the following two changes:
        \begin{enumerate}
            \item The code $C$ specified in \autoref{def: code parameters} has a lower degree.
            \item Each $H_i$ is drawn from a $p$-biased distribution outputting $1$ with probability $p = 2^{-6}$ rather than being uniformly random.
        \end{enumerate}

        As shown in \autoref{lem:code}, the first change weakens the unique-decodability of the dual code $C^\bot$. In particular, we can only guarantee the unique-decodability against a milder distribution $\mathcal{D}_n$ of errors, where each location is scrambled with probability $p = 2^{-6} \leq 0.04$. However, in the analysis of the original Yamakawa--Zhandry problem, the dual code needs to deal with the distribution of errors where each location is scrambled with probability $1/2$. 

        Fortunately, the second change fixes this problem once the biased input distribution is plugged into the analysis.
        Here, we briefly sketch the idea: Recall that the state $\ket{\hat{\phi_i}} = \QFT\ket{\phi_i}$ is generated in the second step of Charlie for each $i \in [n]$. In expectation over $H_i$ drawn from the biased distribution, the state $\ket{\hat{\phi_i}}$ has weight $1-p$ on $0$, and the remaining weight of $p$ is uniformly distributed over $\Sigma \backslash \{0\}$. Therefore, at a high level, before Charlie applies $\Udecode$ in the third step, one can think of the error $e$ is now drawn from the distribution $\mathcal{D}_n$ defined in \autoref{lem:code}, which can be uniquely decoded with high probability by \autoref{lem:code}.
        
        The remaining analysis in \cite{Yamakawa2022} is not affected by these two changes. This concludes the sketch of proof of correctness of our quantum SMP protocol.
    \end{proof}
 
\section{Classical lower bound for \bNC}
\label{sec:lower-bound}
In this section, we prove the  distributional lower bound for the $\bNC$ problem.
\begin{theorem}
\label{thm:binc lb}
    Let $C\subseteq \Sigma^n$ be a $(\zeta, \ell, L)$-list-recoverable code with $\zeta \leq 0.4$ and $L = 2^{o(n)}$. Then $\bNC^C_n$ has randomized communication complexity $\Omega(\ell)$ over the $p$-biased input distribution (\cref{defn: biased oracle}) for any constant $p\in(0,1)$.
\end{theorem} 
In particular, we can later instantiate \cref{thm:binc lb} with the code from \cref{lem:code} (for which we proved the quantum upper bound). We present the proof of \cref{thm:binc lb} in stages. We first (\cref{sec:warmup}) prove the theorem for highly structured communication protocols that we call \emph{subcube protocols}, which are closely related to decision trees. We then (\cref{sec:sub-like recs}) generalize that proof to a more powerful class of \emph{subcube-like protocols}. Finally (\cref{sec:preproc}), we use structure-vs-randomness paradigm~\cite{GPW20, YZ23} to show that any communication protocol can be converted (with modest loss) into a subcube-like protocol.

\paragraph{Input distribution: Simplification.}
Recall that the $p$-biased input distribution assigns all bits \emph{i.i.d.}\ such that the probability of a 1-bit is $p$. We first prove \cref{thm:binc lb} for the uniform distribution, $p=1/2$, and then indicate (\cref{sec:any-p}) how to generalize it to any $p\in(0,1)$.

\subsection{Lower bound for subcube protocols}
\label{sec:warmup}
We start by defining the notion of subcube protocols. We study exclusively \emph{deterministic} protocols---run on a random distribution over inputs---as this is enough to prove a lower bound against randomized protocols~(by Yao's principle).
\begin{definition}[Subcubes]
    A set $X \subseteq \Bool^N$ is a \emph{subcube} if there exists $I \subseteq [N]$ such that $X = \{x\in\Bool^N \colon x_I = a\}$ for some partial assignment $a \in \Bool^{I}$. The \emph{codimension} of $X$ is defined as $|I|$. Moreover, if~$i\in I$ then we say the corresponding variable $x_i$ is \emph{fixed} in $X$. 
\end{definition}

\begin{definition}[Subcube protocols]
A rectangle $R = X \times Y \subseteq\Bool^N\times\Bool^N$ is said to be a \emph{subcube rectangle} if $X$ and $Y$ are subcubes. The codimension of $R$ is then the sum of codimensions of $X$ and $Y$.
A communication protocol $\Pi$ is said to be a \emph{subcube protocol} if for every node in the protocol tree of $\Pi$, the associated rectangle is a subcube rectangle. We denote by $|\Pi|$ the worst-case communication cost of $\Pi$ (maximum length of the transcript).
\end{definition}

The simplest example of a subcube protocol is when the players send individual input bits to each other. For example, Alice can send $x_1\in\Bool$ to Bob, which corresponds to splitting Alice's input domain $\{0,1\}^n$ to two subcubes where $x_1 = 1$ or $x_1 = 0$. Subcube protocols that send one bit at a time correspond precisely to decision trees.

More generally, in each communication round of a subcube protocol, the length of the message can be larger than 1 bit if a node $v$ in the tree has large arity (number of its children). The correspondence between subcube protocols and decision trees breaks when the protocol tree has nodes with arity larger than~2. For example, suppose Alice on input $x\in\Bool^N$ sends the index~$i\in[N]$ of her first 1-bit (or indicates her input is all-0). This is a $\log(N+1)$-bit message that induces a partition of $\Bool^N$ into $N+1$ subcubes. Such a protocol is not efficiently simulated by a shallow subcube protocol sending one bit at a time, that is, a decision tree. (See~\cite{Ambainis16} for more separations between subcube partitions and decision trees.) 

\paragraph{Cleanup.}
We first show that any subcube protocol for $\bNC^C_n$ (or for any communication-$\TFNP$ problem) can be converted into a \emph{zero-error} protocol with bounded codimension at the leaves.

\begin{lemma}[Subcube protocol cleanup] \label{lem:wlog}
Suppose $\Pi$ is a subcube protocol solving $\bNC^C_n$ with error~$\epsilon>0$ over the uniform distribution. Then there exists another subcube protocol $\Pi'$ such that
\begin{enumerate}[label=(C\arabic*),leftmargin=3em]
\item \label{it:one}
$|\Pi'| \leq O(|\Pi|)$.
\item \label{it:two}
$\Pi'$ is \emph{zero-error}: on a uniform random input, it outputs the symbol $\bot$ (``I don't know'') with probability at most $2\epsilon$, and otherwise it outputs a valid solution.
\item \label{it:three}
Every leaf rectangle of $\Pi'$ has codimension $O(|\Pi|/\epsilon)$.
\end{enumerate}
\end{lemma}
\begin{proof}
Since $\Pi$ communicates $|\Pi|$ bits, the protocol tree has at most $2^{|\Pi|}$ leafs. The leafs partition the input domain into at most $2^{|\Pi|}$ rectangles. Denote by $\bm{R}$ the rectangle of the leaf reached when $\Pi$ is run on a uniform random input. A fixed leaf rectangle $R\subseteq\Bool^N\times\Bool^N$ is reached with probability $\Pr[\bm{R}=R]=|R|/2^{2N}\leq 2^{-\codim(R)}$. Letting $\H(\,\cdot\,)$ denote Shannon entropy, we have
\[
\E[\codim(\bm{R})]\leq \E[\log(1/\Pr[\bm{R}=R])] =
\H(\bm{R}) \leq |\Pi|.
\]
By Markov's inequality, $\Pr[\codim(\bm{R})\leq |\Pi|/\epsilon]\geq 1-\epsilon$. If we reach a node fixing more than $|\Pi|/\epsilon$ indices, we modify $\Pi$ to output $\bot$. This produces a protocol $\Pi'$ satisfying \ref{it:one} and \ref{it:three}.

Finally, we assume that the protocol $\Pi$ ends with Bob outputting the solution. Alice can then check if his solution is valid for her input, and send a 1 if yes. If she sends a 0, Bob outputs $\bot$. This means our protocol can be converted into one which is zero-error via this round of interaction which verifies the solution outputted by Bob. In the modified protocol, the success probability is the same and we can assume $\Pi$ always outputs a correct solution or $\bot$. The total probability of outputting $\bot$ is the error probability of the original protocol plus the probability of the leaf codimension being greater than $|\Pi|/\epsilon$. Thus, the probability is $2\epsilon$, which satisfies \ref{it:two}.
\end{proof}

\begin{proof}[Proof of \cref{thm:binc lb} for subcube protocols]
By \cref{lem:wlog} we may assume w.l.o.g.\ that $\Pi$ is zero-error and has bounded codimension at the leaves. For the sake of contradiction, assume that $|\Pi| \leq o(\ell)$. Sample a uniform random input $\bm{H}\sim\{0,1\}^N\times\{0,1\}^N$ and run the protocol to get a random transcript $\Pi(\bm{H})$. In this run, we encounter a sequence of nodes $\bm{v}_1,\ldots,\bm{v}_d$ of the protocol tree of $\Pi$ where $\bm{v}_1$ is the root node, $\bm{v}_d$ is a leaf, and $d$ is the number of communication rounds.

We call a codeword $c \in C$ \emph{dangerous} in a partial assignment $\rho \in\{0,1,*\}^N\times \{0,1,*\}^N$ if at least~$0.4n$ bits of $c$ are fixed by $\rho$. In communication round $i$, we denote by $\bm{Q}_i \subseteq C$ the set of codewords that are dangerous for the partial assignment corresponding to the subcube rectangle at node $\bm{v}_i$. We observe three properties of $\bm{Q} \coloneqq (\bm{Q}_1, \ldots \bm{Q}_d)$:

\begin{enumerate}[label=(\roman*)]
\item
\label{it:1}
${\bm{Q}}$ is monotone: $\bm{Q}_i \subseteq \bm{Q}_{i+1}$ for all $i$.
\item
\label{it:2}
If $\Pi(\bm{H})$ outputs $c\neq \bot$, then $c\in\bm{Q}_d$. 
\item \label{claim:bound}
$|\bm{Q}_d| \leq L$.
\end{enumerate}
Here \ref{it:1} is immediate. Item \ref{it:2} follows because $\Pi$ is zero-error. To see \ref{claim:bound}, note that since $|\Pi|\leq o(\ell)$ we have from \cref{lem:wlog} that at most $O(|\Pi|)\leq o(\ell)$ bits of $\bm{H}$ are fixed. By the $(0.4, \ell, L)$-list-recoverability of the code $C$, we obtain \ref{claim:bound}.
Using \ref{claim:bound} we can write $\bm{Q}_d = \{\bm{c}_1, \bm{c}_2, \ldots, \bm{c}_L\}$ where every $\bm{c}_i\in C$ is a random variable.

    \begin{claim}
    \label{lem:subcube bound}
        For every $j \in [L]$ we have
        $\Pr [\bm{H}(\bm{c}_j) = 0^n] \leq 2^{-\Omega(n)}$.
    \end{claim}
    \begin{proof}
        Fix $j\in[L]$ and consider the first node $\bm{v}\in\{\bm{v}_1,\ldots,\bm{v}_d\}$ where $\bm{c}_j$ becomes dangerous. That is, $\bm{v}=\bm{v}_i$ such that $\bm{c}_j\in\bm{Q}_i$ but $\bm{c}_j\notin\bm{Q}_{i-1}$. Let $R_v$ denote the rectangle corresponding to node~$v$. Note that for $v\in\operatorname{supp}(\bm{v})$, the events ``$\bm{v} = v$'' and ``$\bm{H} \in R_v$'' are the same: monotonicity~\ref{it:1} states that once a codeword $c\in C$ becomes dangerous at node $v$, it will stay dangerous for all leafs in $v$'s subtree. Suppose Alice sent a message at $\bm{v}$'s parent $\bm{v}_{i-1}$. Since $\bm{c}_j\notin \bm{Q}_{i-1}$, Bob has at most $0.4n$ bits of $\bm{c}_j$ fixed at $R_{\bm{v}_{i-1}}$, and hence also at $R_{\bm{v}}$. Because of these $0.1n$ unfixed bits of $\bm{c}_j$ on Bob's side, we have
        \begin{equation}\label{eq:unfixed}
        \Pr[\bm{H}(\bm{c}_j) = 0^n \mid \bm{v}=v] ~=~ \Pr[\bm{H}(\bm{c}_j) = 0^n \mid \bm{H}\in R_v] ~\leq~ 2^{-0.1n} ~\leq~ 2^{-\Omega(n)}.
        \end{equation}
        We now calculate
        \begin{align*}
            \Pr[\bm{H}(\bm{c}_j) = 0^n] &= \sum_{v}\Pr[\bm{H}(\bm{c}_j)
            = 0^n \mid \bm{v}=v] \cdot \Pr[\bm{v}=v] \\
            &\leq 2^{-\Omega(n)} \sum_{v} \Pr[\bm{v}=v] \tag{by~\eqref{eq:unfixed}}\\
            &= 2^{-\Omega(n)}. \qedhere
        \end{align*}
    \end{proof}
    Using \cref{lem:subcube bound}, we can now show that $\Pi$ errs with high probability, which finishes the proof.
    \begin{align*}
         \Pr[\Pi(\bm{H}) \text{ is correct}] &\leq \Pr[\exists c\in\bm{Q}_d\colon \bm{H}(c) = 0^n]   \tag{by \ref{it:2}}\\
         &\leq \sum_{j \in [L]} \Pr[\bm{H}(\bm{c}_j) = 0^n]  &\tag{union bound}\\
         &\leq L \cdot 2^{-\Omega(n)} &\tag{\cref{lem:subcube bound}}\\
         &\leq 2^{o(n)} \cdot 2^{-\Omega(n)} \leq o(1). \tag*{\qedhere}
    \end{align*}
\end{proof}

\subsection{Subcube-like protocols}
\label{sec:sub-like recs}

We will now prove \cref{thm:binc lb} for a more general class of \emph{subcube-like protocols}. To set this up, we introduce the notion of \emph{dense} random variables, which is a notion of pseudorandomness first considered in~\cite{Goos2016}. We first recall the definition of min-entropy.

\begin{definition}[Min-entropy]
    Given a random variable $\bm{X}\in\cX$, its min-entropy is defined to be
    $$\Hmin(\bm{X}) \coloneqq \log \frac{1}{\max_{x \in \cX} \Pr[\bm{X} = x]}$$
\end{definition}
For example, note that $\Hmin(\bm{X}) \geq t$ iff we have $\Pr[\bm{X} = x] \leq 2^{-t}$ for all $x$. 

\begin{definition}
        A random variable $\bm{X}\in\Bool^N$ is said to be \emph{$\gamma$-dense} if for every $I \subseteq [N]$, we have
    $$\Hmin(\bm{X}_I) \geq \gamma|I|$$
\end{definition}

\begin{definition}[Subcube-like rectangle]
\label{defn:struc rec}
    A rectangle $R = X\times Y$ is said to be \emph{$\gamma$-subcube-like} with respect to $(I, J)$ if (here we denote $\bar{I}\coloneqq [N]\setminus I$)
    \begin{itemize}
        \item $X_{I}$ and $Y_{J}$ are fixed strings.
        \item $\bm{X}_{\bar{I}}$ and $\bm{Y}_{\bar{J}}$ are $\gamma$-dense where $\bm{X}\sim X$ and $\bm{Y}\sim Y$ are sampled uniformly at random.
    \end{itemize}
\end{definition}

Similar to subcubes, the codimension of a subcube-like rectangle w.r.t.\ $(I,J)$ is said to be $|I|+|J|$. A communication protocol $\Pi$ is said to be a \emph{subcube-like protocol} if for every node in the protocol tree, the associated rectangle is subcube-like. We can now re-prove the cleanup lemma (\cref{lem:wlog}) for subcube-like protocols; the original proof works verbatim.

\begin{lemma}[Subcube-like protocol cleanup] \label{lem:wlog-like}
Suppose $\Pi$ is a $\gamma$-subcube-like protocol solving $\bNC^C_n$ with error $\epsilon>0$ over the uniform distribution. Then there exist another $\gamma$-subcube-like protocol $\Pi'$ such that \cref{it:one,it:two,it:three} are satisfied. \hfill $\qed$
\end{lemma}

We are now ready to prove that \cref{thm:binc lb} holds for $0.8$-subcube-like protocols.

\begin{proof}[Proof of \cref{thm:binc lb} for subcube-like protocols]
The argument is nearly identical to that for subcube protocols. Define $\bm{Q} \coloneqq (\bm{Q}_1, \ldots \bm{Q}_d)$ as before for the candidate subcube-like protocol $\Pi$. The only crucial difference is in \cref{lem:subcube bound}: We show that it continues to hold (albeit with a smaller constant hidden by the $\Omega$-notation). Indeed, if $\bm{c}_j$ becomes dangerous for the first time at~$\bm{v}=v$, then Bob (or Alice) has at least $0.1n$ unfixed bits of $\bm{c}_j$ in his input set $Y$ where $R_v=X\times Y$. These unfixed bits come from a $0.8$-dense distribution $\bm{Y}\sim Y$ where the min-entropy of the marginal distribution of these unfixed bits is at least~$0.8\cdot 0.1n=0.08n$. This implies $\Pr[\bm{H}(\bm{c}_j) = 0^n \mid \bm{H}\in R_v] \leq 2^{-0.08n}\leq 2^{-\Omega(n)}$ showing that the bound~\eqref{eq:unfixed} continues to hold, which re-proves \cref{lem:subcube bound}. We can finally conclude that $\Pr[\Pi(\bm{H}) \text{ is correct}] = o(1)$ using the same calculation as before.
\end{proof}

\subsection{From general protocols to subcube-like protocols}
\label{sec:preproc}

We will now prove \cref{thm:binc lb} for general protocols. The proof is by a reduction to the case of subcube-like protocols using the following theorem.

\begin{theorem}[Protocol preprocessing]
\label{cor:preproc}
    Suppose $\Pi$ is a protocol solving some communication problem with error $\epsilon>0$ over the uniform distribution. Then for any constant $\gamma<1$ there exist a $\gamma$-subcube-like protocol $\Pi'$ such that
    \begin{enumerate}
     \item \label{it:cc-bound}
     $|\Pi'|\leq O(|\Pi|/\epsilon)$
     \item \label{it:err-bound}
     $\Pi'$ has error $2\epsilon$ over the uniform distribution.
    \end{enumerate}
\end{theorem}

Simply put, \cref{cor:preproc} states that any  protocol can be assumed to be subcube-like without loss of generality.
A similar theorem was proved in~\cite{YZ23, WYZ23}, but instead of a bound on the communication cost (\cref{it:cc-bound}), their version produced a subcube-like protocol with a bound on the codimension of the leafs. This is a slightly weaker guarantee than \cref{it:cc-bound}: Recall from \cref{lem:wlog}\ref{it:three} that bounded communication cost implies a bound on the leaf codimension. We prefer our version in \cref{cor:preproc} as it is more elegant and potentially easier to use in subsequent work. 

\paragraph{Proof sketch.}
\cref{cor:preproc} is only a slight strengthening of the techniques in~\cite{GPW20,YZ23, WYZ23} and so we are content with only sketching the crucial differences. To begin, we assume the protocol $\Pi$ is in the normal form where the players send one bit in each communication round. Protocol $\Pi'$ is now defined by \cref{alg:main}. The only difference to~\cite{YZ23, WYZ23} is that we send the index $i$ encoded using a Huffman code. That is, we use the following fact.
\begin{fact}[\cite{Huffman52}]
\label{lem:huffman}
    Let $\bm{k}$ be any random variable. There exist an encoding function~$C(\,\cdot\,)$ (called the Huffman code) such that $\E[|C(\bm{k})|] \leq \H(\bm{k}) + 1$.
\end{fact}

Otherwise the algorithm is the one stated in \cite{YZ23, WYZ23}. It crucially relies on the \emph{density-restoring partition} from \cite{GPW20}, which states:

\begin{lemma}[Density-restoring partition \cite{GPW20,CFKMP21}]
\label{lem:density-restoring-partition}Let $\bm{X}\in\Bool^N$ be a random variable and let $0\leq \gamma<1$ be a constant. Then there exists a partition 
\[
\cX\eqdef\cX^{1}\cup\dots\cup\cX^{\ell}
\]
where every $\cX^{j}$ is associated with a set $I_{j}\subseteq\left[n\right]$
and a value $x_{j}\in \Bool^{I_{j}}$ such that:
\begin{itemize}
\item $(\bm{X}_{I_{j}}\texttt{\ensuremath{\mid}}\bm{X}\in\cX^{j})$ is fixed to~$x_{j}$.
\item $(\bm{X}_{\bar{I}_{j}}\texttt{\ensuremath{\mid}}\bm{X}\in\cX^{j})$ is
$\gamma$-dense.
\item If we pick $\bm{i}$ at random according to $\Pr[\bm{i}=i]=|\calX^i|/|\calX|$, then the expected codimension of $\calX^{\bm{i}}$ is $n - \Hmin(\bm{X}) + O(1)$.
\end{itemize}
\end{lemma}

The proof of this lemma is by analysing a greedy algorithm that constructs the $\calX^j$s one by one. We refer the reader to \cite{GPW20,CFKMP21} for a proof. We now state the key fact about~\cref{alg:main}: over the uniform input distribution, the the probabilities of typical \emph{transcripts} of $\Pi'$ (concatenation of $\Pi'$s messages) are similar to those of $\Pi$. 
We write $\bm{\Pi'}\coloneqq \Pi'(\bm{X}, \bm{Y})$ for the transcript generated on a uniform random input $(\bm{X}, \bm{Y}) \sim \Bool^N \times \Bool^N$. The following is a version of \cite[Lemma~3.3]{YZ23}. Their original statement bounds the expected codimension of the leaf $\ell$ reached, but the same proof implicitly bounds $\log(1/\Pr[\bm{\Pi'}=\ell])$, whose average value is~$\H(\bm{\Pi'})$ by definition. (Note that our use of Huffman encoding only relabels messages, which does not change the entropy of the transcript.)

\begin{lemma}[\cite{YZ23}]\label{lem: fixed equal cc} $\H(\bm{\Pi'}) \leq O(| \Pi |)$.
\end{lemma}

\begin{algorithm}[t]
    \caption{Protocol $\Pi'$ on input $(x,y)$}
    \begin{algorithmic}
    \label{alg:main}
        \State initialize:~~ $v \leftarrow \text{root of $\Pi$}$,~~ $X\times Y=\Bool^N \times \Bool^N$,~~ $I,J\leftarrow \emptyset$
        \While {{$v$ is not at leaf level}}~~
        {\color{Maroon}
        {\bf[}\,invariant: $X\times Y$ is $\gamma$-subcube-like w.r.t. $(I, J)$\,{\bf]}
        }
        \State let $v_0$, $v_1$ be the children of $v$

        \State {\bf Suppose} Alice sends a bit at $v$ \emph{(otherwise swap roles of $X$ and $Y$)}
        \State let $X =X^0\cup X^1$ be the partition according to Alice's message function at $v$
        \State update $X \leftarrow X^b$ where $b$ is such that $x\in X^b$
        \State let $X=\bigcup_i X^i$ be a density-restoring partition for $X$ (with associated sets $I_i$)
        \State Let $C(\,\cdot\,)$ be the Huffman encoding of $\bm{k}$ where $\Pr[\bm{k} = i] = |X_i|/|X|$ 
        \State update $X\leftarrow X^i$ and $I\leftarrow I \cup I_i$ where $i$ is such that $x\in X^i$
        \State {\bf Alice sends} $(b,C(i))$ to Bob
        \State update $v \leftarrow v_b$
       \EndWhile
        \State {\bf Output} the same as leaf $v$
    \end{algorithmic}
\end{algorithm}

We can now verify \cref{it:cc-bound,it:err-bound}. Denote by $\bm{v}_1,\ldots,\bm{v}_d$ the nodes of the protocol tree of $\Pi'$ encountered when running it on a uniform random input. That is, $\bm{v}_1$ is the root node, $\bm{v}_d$ is a leaf, and $d$ is the number of communication rounds. We can thus calculate,
\begin{align*}
\E[|\bm{\Pi'}|]
&\textstyle =\sum_{k\in[d]} \E[|(\bm{b},C(\bm{i}))\text{ at } \bm{v}_k|] \\
&\textstyle \leq
d + \sum_k \E[|C(\bm{i})\text{ at } \bm{v}_k|] \\
&\textstyle \leq d + \sum_k (\H(\bm{i} \text{ at } \bm{v}_k)+1) \tag{\cref{lem:huffman}} \\
&\leq 2d + \H(\bm{\Pi'}) \\
&\leq O(|\Pi|). \tag{\cref{lem: fixed equal cc}}
\end{align*}
If the communication cost exceeds $O(|\Pi|)$ by a factor larger than $1/\epsilon$, we make Alice and Bob abort and output $\bot$. By Markov's inequality, this increases the probability of~$\bot$ by at most $\epsilon$. For the new protocol, we have \cref{it:cc-bound} by definition and \cref{it:err-bound} since the new error is at most $\epsilon + \epsilon = 2\epsilon$. This concludes the proof sketch of~\cref{cor:preproc}.

\paragraph{Remark.} Even though we assumed $\Pi$ sends one bit per communication round, $\Pi'$ can send long \emph{variable-length} messages. This is important in order to ensure every node in $\Pi'$ is subcube-like.

\subsection{Proof for any \texorpdfstring{$p\in(0,1)$}{p in (0,1)}}
\label{sec:any-p}
To extend the proof of \cref{thm:binc lb} to $p$-biased distributions, we will do the analysis over a modified definition of the input. For any $p = 2^{-k}$ for some integer $k$, define $\bias{H} \defeq (\textup{AND}^{N}_k(\unf{H}_1),\textup{AND}^{N}_k(\unf{H}_2))$ where $\unf{H} \defeq \unf{H}_1, \unf{H}_2 \in \Bool^{kN}$. That is, $H$ is defined by dividing the input $\unf{H}$ into blocks of size $k$ and performing a logical \textup{AND} on those bits. Sampling $\bm{\unf{H}}$ uniformly from $\Bool^{kN} \times \Bool^{kN}$ then results in a $p$-biased distribution for $\bm{\bias{H}}$. Thus, doing the analysis for this modified problem over the uniform distribution gives us a lower bound for $\bNC^C$ over the $p$-biased distribution.

Since our input is now a uniform distribution, we can assume we have a subcube-like protocol by \cref{cor:preproc}. Note that for any $p_2 > p_1$, \bNC is easier with the $p_1$-biased distributions than $p_2$-biased distributions. If Alice and Bob are  given a sample $H$ from a $p_1$-biased distribution, then Alice and Bob can flip each $0$-bit independently to $1$ with probability $(p_2-p_1)/(1-p_1)$. The result is another random variable $H'$ which is $p_2$-biased. If they run the algorithm on $H'$ and get a solution~$c$, then that $c$ is also valid for $H$. Since for any $p \in (0,1)$ there is an integer $k$ such that $p \geq 2^{-k}$, it suffices to prove the lower bound for such distributions. We claim that \cref{thm:binc lb} also continues to hold for subcube-like protocols when we have a $p$-biased distribution for any constant $p = 2^{-k}$.

\begin{proof}[Proof of \cref{thm:binc lb} for $p$-biased distributions]
We sample $\bm{\unf{H}} \sim \Bool^{kN} \times \Bool^{kN}$, which induces a $p$-biased distribution on ${\bm{H}}$, where $p = 2^{-k}$. We preprocess the protocol over $\bm{\unf{H}}$ using \cref{cor:preproc}. Define ${\cal Q} = \bm{Q}_1, \ldots \bm{Q}_d$ the same way for the candidate subcube-like protocol $\Pi$ over this distribution. Note that \cref{lem:subcube bound} continues to hold.
For any $\bm{c}_j$, if it becomes dangerous for the first time at $\bm{v}$, then Bob has at least $0.1n$ unfixed bits in his input. Since Bob's rectangle is subcube-like, the underlying bits of $\unf{H}$ come from a $0.8$-dense distribution. The min-entropy of the marginal distribution for these bits is at least $0.08kn$.
This implies $\Pr[\bm{\bias{H}}(\bm{c}_j) = 0^n \mid \bm{\unf{H}} \in R_{\bm{v}}] \leq 2^{-0.08kn} = 2^{-\Omega(n)}$.
We can then conclude that $\Pr[\Pi(\bm{\unf{H}}) \text{ is correct}] = o(1)$ using the same calculation as before.
\end{proof}

\section{Converting to a total relation}
\label{sec:total}

So far, we have shown that $\bNC$ achieves an exponential separation between quantum SMP and classical two-way communication \emph{over a biased input distribution}. In this section, we convert $\bNC$ into a total problem and prove our main theorem. 

\ThmMain*

\subsection{Construction} 
Following the approach of Yamakawa and Zhandry~\cite[Section~6]{Yamakawa2022} towards obtaining a separation for \emph{total} problems, we consider a total problem variant of $\bNC$, to which we refer to as Total Bipartite NullCodeword (\TbNC). At a high level, we allow Alice and Bob to pick a common hash function from a $\lambda(n)$-wise independent hash function family, and then $\textup{XOR}$ it with their own input. The quantum easiness still holds for any $\lambda(n)$-wise independent random inputs. However, the freedom of choosing a common hash function also makes the problem too easy, so we further require Alice and Bob to solve $t$ instances of \bNC while using the same common hash function.

To provide a formal definition, we first introduce some notation.
Recall that we can simulate the $p$-biased distribution (\cref{defn: biased oracle}), for $p = 2^{-6}$, by splitting each bit into the conjunction of six uniformly random bits. For any function $\bias{f}\colon \mathcal{X} \rightarrow \Bool^{m}$, where $\mathcal{X}$ is any input domain, we use notation $\unf{f}\colon \mathcal{X} \rightarrow \Bool^{6m}$ for the expanded version. Though, we often give $\unf{f}\colon \mathcal{X} \rightarrow \Bool^{6m}$ first, and then define $\bias{f}\colon \mathcal{X} \rightarrow \Bool^{m}$ by taking the $\textup{AND}$ of every six consecutive bits of the output of $\unf{f}$.

Fix $C$ to be the code defined in \autoref{def: code parameters}.
We abuse notation and write $\unf{h}\colon \Sigma^n \rightarrow \{0,1\}^{6n}$ as the concatenation of $\unf{h}(\cdot, 1), \ldots, \unf{h}(\cdot, n)$.
We take $\lambda(n) = n^2$ and assume $\{\unf{h}_k\}$ is implemented by the standard low-degree polynomial construction. See, e.g., \cite[Section 3.5.5]{Vadhan12psr}. With this construction, $|\mathcal{K}| = 2^{r}$ with $r = \poly(n)$.

\begin{definition}[Total Bipartite NullCodeword, \TbNC]
    \label{defn:tbinc}
    Take $t \coloneqq nr$. Let $\unf{H}^{(1)} \ldots, \unf{H}^{(t)}$ be $t$ mappings $\Sigma^n\to\Bool^{6n}$, where each $\unf{H}^{(i)}$ is the concatenation of $n$ mappings $\unf{H}^{(i)}_1, \ldots, \unf{H}^{(i)}_n$ of type $\Sigma\to\Bool^{6}$.
    Let $\unf{H}^{(i)}_A, \unf{H}^{(i)}_B\colon \Sigma^{n/2} \rightarrow \Bool^{3n}$ be the first half and the second half of $\unf{H}^{(i)}$.

    $\TbNC$ is defined as the following communication problem: Alice and Bob are given $t$ mappings $(\unf{H}^{(1)}_A, \ldots, \unf{H}^{(t)}_A)$ and $(\unf{H}^{(1)}_B, \ldots, \unf{H}^{(t)}_B)$, respectively, and they need to find a key $k \in \mathcal{K}$ and $t$ strings $x^{(1)}, \ldots, x^{(t)}$ such that
    \[
    \forall i \in [t]:\qquad
    x^{(i)} \in C_n \quad \textrm{and} \quad \bias{H}^{(i)}(x^{(i)}) \oplus \bias{h}_k(x^{(i)}) = 0^n.\]
\end{definition}

\paragraph{Remark (totality).}
    The totality of $\TbNC$ is implied by the fact that for any input, the quantum SMP protocol can always find a solution with high probability, which will be proved later. A more direct proof can be shown as follows: Fixing any input $\unf{H}^{(1)}, \ldots, \unf{H}^{(t)}$. For each $i \in [t]$, by the $\lambda(n)$-wise independence of $\unf{h}_k$, it holds that 
    \[ \Pr_{k \sim \mathcal{K}} [\{x \in C_n : \bias{H}^{(i)}(x^{(i)}) \oplus \bias{h}_k(x^{(i)}) = 0^n \} = \emptyset] = 2^{-\Omega(n)}. \] 
    Since $t = \omega(n)$, there exists a key $k$ such that the set of codewords mapped to $0^n$ is non-empty.

\subsection{Worst-case quantum easiness for \TbNC}

We here present a proof sketch for the quantum easiness of \TbNC.

\begin{lemma}\label{lem: quantum easiness of TBNC}
    $\TbNC$ has a quantum SMP protocol using $\poly(n)$ qubits over the worst-case input.
\end{lemma}

\begin{proof}[Proof Sketch]
    As in \cite{Yamakawa2022}, the query algorithm for the total version of the \NC ($\TYZ^C$)
    will first randomly pick a key $k \sim \mathcal{K}$, and then solve $t$ copies of $\YZ^C$, where the $i$-th $\YZ^C$ instance is given by oracle $\bias{H^{(i)}} \oplus \bias{h}_k$. Let us call this algorithm $\AlgT$.
    Similar to the proof of \autoref{thm: quantum easiness}, the algorithm $\AlgT$ can also be captured by the same two-stage framework, \textbf{Query} and \textbf{Process}. 
    Therefore, our SMP protocol for $\TbNC$ can simulate it in the same manner: Alice and Bob simulate the \textbf{Query} stage on their own inputs and send their results to Charlie; Charlie then simulates the \textbf{Process} stage and output the solution. 
    See \autoref{alg:smp_protocol for TbNC} for the details of the quantum SMP protocol for $\TbNC$. 

    Now we have to show that for the code $C$ specified in \autoref{def: code parameters}, the algorithm $\AlgT$ correctly solves $\TYZ^C$. It was shown in \cite[Lemma 6.10 using Lemma 2.5]{Yamakawa2022, Zhandry12} that since $\AlgT$ uses the algorithm $\AlgYZ$ for solving $\YZ^C$ as a sub-procedure, the correctness of $\AlgYZ$ (in the average case) would imply the correctness of $\AlgT$ in the worst case. Recall that the correctness of $\AlgYZ$ for $\YZ^C$ is already proven in \autoref{thm: quantum easiness}, we conclude \autoref{lem: quantum easiness of TBNC}.
\end{proof}

\begin{algorithm}[hp]
\caption{Quantum SMP Protocol for $\TbNC$}\label{alg:smp_protocol for TbNC}
\begin{algorithmic}[0]

\State \textbf{Notation:} 
 
    \State For each $i \in [t], j \in [n]$, define the states
    \[
        \ket{\phi^{(i)}_j} \propto \sum_{e \in \Sigma} \ket{e}\ket{H^{(i)}_j(e)}, \; \ket{\overline{\phi}^{(i)}_j} \propto \sum_{e \in \Sigma} \ket{e}\ket{H^{(i)}_j(e) \oplus h_k(e, j)}
    \]
    and define the set 
    \[T^{(i)}_j = \{e \in \Sigma: H^{(i)}_j(e) \oplus h_k(e, j) = 0 \}.\]

\State Let $\Uadd, \Udecode$ be unitaries defined as follows:
\[
\ket{x}\ket{e} \xrightarrow{\Uadd} \ket{x}\ket{x+e} \xrightarrow{\Udecode} \ket{x - \decode(x+e)}\ket{x+e}.
\]

\vspace{-1.5mm} 
\Statex 
\noindent \rule{\linewidth}{0.4pt}

\State \textbf{Alice:} 
\State For each $i \in [t]$:
\begin{enumerate}
    \item Prepare $n$ copies of $\ket{\phi^{(i)}_1}, \ldots, \ket{\phi^{(i)}_{n/2}}$ using $H^{(i)}_1, \ldots, H^{(i)}_{n/2}$.
    \item Send those copies of $\ket{\phi^{(i)}_1}, \ldots, \ket{\phi^{(i)}_{n/2}}$ to Charlie.
\end{enumerate}

\vspace{-1.5mm}
\Statex

\State \textbf{Bob:}
\State For each $i \in [t]$:
\begin{enumerate}
    \item Prepare $n$ copies of  $\ket{\phi^{(i)}_{n/2+1}}, \ldots, \ket{\phi^{(i)}_n}$ using $H^{(i)}_{n/2+1}, \ldots, H^{(i)}_n$.
    \item Send those copies of $\ket{\phi^{(i)}_{n/2+1}}, \ldots, \ket{\phi^{(i)}_n}$ to Charlie.
\end{enumerate}

\vspace{-1.5mm}
\Statex

\State \textbf{Charlie:}
    \State Randomly draw a key $k \sim \mathcal{K}$ and prepare $t$ copies of the state
        $\ket{\psi} \propto \sum_{x \in C} \ket{x}.$
    \State \For{$i \in [t]$}
    \begin{enumerate}
        \item For each $j \in [n]$, generate the state $\ket{\overline{\phi}^{(i)}_j}$ using $\ket{\phi^{(i)}_j}$ and then measure the second register of it. If the measurement returns $0$, then Charlie successfully generates the state
        \[ \ket{\phi^{k}_j} \propto \sum_{e \in T^{(i)}_j} \ket{e}; \]
        otherwise, retry with another copy of $\ket{\phi^{(i)}_j}$.
        
        \item Apply $\QFT$ to $\ket{\phi^k} \coloneqq \ket{\phi^k_1} \otimes \ldots \otimes \ket{\phi^k_n}$ and $\ket{\psi}$. Now, Charlie holds the state:
        \[
        \ket{\eta} \coloneqq \QFT\ket{\psi} \otimes \QFT\ket{\phi^k}.
        \]
        \item Apply $(I \otimes \QFT^{-1}) \Udecode \Uadd$ to $\ket{\eta}$.
        \item Measure the second register and store the measurement result as $x^{(i)}$.
    \end{enumerate}
    \EndFor
    \State Output the solution $(k, x^{(1)}, \ldots, x^{(t)})$.
\end{algorithmic}
\end{algorithm}

\subsection{Average-case classical lower bound for \TbNC}

The average-case classical lower bound follows via a union bound after applying a direct product theorem. We shall use the following direct product theorem for randomized communication complexity by Braverman, Rao, Weinstein, and Yehudayoff~\cite{braverman2013}.

\begin{theorem}[\cite{braverman2013}]\label{thm: direct product}
    Let $f$ be any communication problem (function or relation), and $\mu$ be a product distribution of $f$'s input. Denote the maximum success probability of a classical two-way communication protocol for solving $f$ by $\suc(f, \mu, C)$.
    Let $f^n$ denote the problem of solving $n$ copies of $f$ and $\mu^n$ denote the product distribution on $n$ inputs.
    
    Then, if $T \log^2 T = o(nC)$ and $\suc(f, \mu, C) \leq 2/3$, it holds that $\suc(f^n, \mu^n, T) \leq 2^{-\Omega(n)}$.
\end{theorem}

\paragraph{Remark (functions vs relations).} \autoref{thm: direct product} was originally proven for boolean functions $f$ in \cite{braverman2013}. However, their technique does not depend on whether there is a unique solution or many valid solutions. Specifically for product distributions, the technique is to
\begin{enumerate}[label=(\roman*)]
\item extract a protocol with low \emph{external information} cost, and then
\label{it: direct product 1}
\item compress the protocol using ideas from \cite{BBCR13}.
\label{it: direct product 2}
\end{enumerate}

Step \ref{it: direct product 1} is easily seen to be true regardless of the number of solutions. Step \ref{it: direct product 2} works for relations because the key idea is to simulate the transcript of the protocol, which is oblivious to whether we are solving a function or relation.

Now we are ready to establish the classical lower bound of $\TbNC$ in the average case.

\begin{lemma}
    \TbNC requires $2^{n^{\Omega(1)}}$ communication for any classical two-way randomized communication protocol when the input is drawn from the uniform distribution.
\end{lemma}

\begin{proof}
    Suppose that the inputs $(\unf{H}^{(1)}_A, \ldots, \unf{H}^{(t)}_A)$ and $(\unf{H}^{(1)}_B, \ldots, \unf{H}^{(t)}_B)$ are drawn from the uniform distribution. By definition, for each $i \in [t]$, $\bias{H}^{(i)}_A$ and $\bias{H}^{(i)}_B$ are drawn from the $p$-biased distribution.
    
    Now consider a fixed hash function $\unf{h}_k$ where $k \in \mathcal{K}$. Note that for each $i \in [t]$, $\bias{H}^{(i)}_A \oplus h_k$ and $\bias{H}^{(i)}_B \oplus h_k$ are also drawn from the $p$-biased distribution, and they can be viewed as $t$ copies of input to \bNC. By \autoref{thm:binc lb}, we know that there exists a number $C = 2^{n^{\Omega(1)}}$ such that $\suc(\bNC, \mu, C) \leq 2/3$, where $\mu$ is the $p$-biased distribution. Then, we apply \autoref{thm: direct product}, it holds that $\suc(\bNC^t, \mu^t, C) \leq 2^{-\Omega(t)}$.
    
    Finally, by a union bound over all $k \in \mathcal{K}$, the success probability of a communication protocol for $\TbNC$ using $C$ bits is at most \[|\mathcal{K}| \cdot 2^{-\Omega(t)} = 2^r \cdot 2^{-\Omega(nr)} = 2^{-\Omega(n)},\]
    since $|\mathcal{K}| = 2^r$ and $t = nr$.
\end{proof}

\section{Future directions}
We highlight three directions for future investigation.

\paragraph{Total boolean functions.}
The most immediate problems left open by our work are showing any exponential separations between classical and quantum communication for \emph{total boolean functions}, or proving that they are always polynomially related. We do not expect any modification of the Yamakawa--Zhandry problem to work since it is a \emph{search} problem with inherently exponential number of solutions.

\paragraph{Multiparty communication.}

Beyond the 2-party model, it would be interesting to show that $k$-party quantum NOF (Number on Forehead) communication can be exponentially separated from $k$-party randomized NOF communication. We expect that a natural 3-party version of the problem we study achieves this separation, where each player gets $2/3$ of the input oracles. However, in the NOF model we lack techniques to lower bound randomized communication which don't also lower bound quantum communication (see \cite{LSS09}). Note that lower bounds in the 2-party model already imply lower bounds in the $k$-party NIH (Number in Hand) model, and it is easy to verify that our quantum upper bound does indeed work in the NIH model when the inputs to \NC are split into $k$ equal parts among the $k$ parties. Hence, our result gives a separation in the $k$-party NIH model.

\paragraph{NISQ separation.} The era of fault-tolerant quantum computers has not arrived yet, but we do have so-called NISQ (noisy intermediate-scale quantum) devices. Chen et al.~\cite{Chen2023} defined NISQ as a computational class and studied its complexity. Can we find a quantum--classical communication separation where the quantum upper bound is in an analogue of their class?

\appendix
\section{Correctness of the quantum SMP protocol for \bNC}\label{sec: full proof of quantum easiness}


In this section, we prove the correctness of our quantum SMP protocol described in \cref{alg:smp_protocol}.

\begin{lemma}\label{lem: correctness of SMP protocol}
    When each $H_i$ is drawn from a $p$-biased distribution distribution, Charlie outputs a valid solution of $\bNC^C$ with high probability.
\end{lemma}

Recall that for each $i \in [n]$,
    \[\ket{\phi_i} \propto \sum_{e_i \in \Sigma: \; H_i(e_i) = 0} \ket{e_i}, \;\ket{\phi} \coloneqq \ket{\phi_1} \otimes \ldots \otimes \ket{\phi_{n}},\]
    and \[
    \ket{\psi} \propto \sum_{x \in C} \ket{x}. \]

After receiving the state $\ket{\phi}$, Charlie is essentially solving the $\YZ^C$ problem without any additional queries to $H_1, \ldots, H_n$. Recall that in the original algorithm $\AlgYZ$ for $\YZ^C$, the first part is preparing the state $\ket{\phi}$ using queries to $H_1, \ldots, H_n$, and the second part is processing $\ket{\phi}$ without queries. Charlie's algorithm is the same as the second part of $\AlgYZ$. Therefore, it suffices to check that the original analysis of the $\AlgYZ$ in \cite{Yamakawa2022} still holds for the new code $C$ and biased input distribution. 

\paragraph{Analysis of the Yamakawa--Zhandry Algorithm.}

We follow the original analysis in \cite{Yamakawa2022} and highlight where it has to be changed. We will use $\negl(n)$ to denote the class of functions which are asymptotically smaller than an inverse-polynomial in $n$, which is the standard notation in cryptographic literature.
The main technical lemma in \cite{Yamakawa2022} is the following.

\begin{lemma}[Lemma 5.1 in \cite{Yamakawa2022}]\label{lem: main technical}


        Let $\ket{\psi}$ and $\ket{\phi}$ be quantum states on a quantum system over an alphabet $\Sigma=\FF_q^m$ written as 
        \begin{align*}
            &\ket{\psi}=\sum_{\vx \in \Sigma^n}V(\vx)\ket{\vx}, \ket{\phi}=\sum_{\ve \in \Sigma^n}W(\ve)\ket{\ve}.
        \end{align*}
        Let $F\colon\Sigma^n \rightarrow \Sigma^n$ be a function. 
        Let $\good\subseteq \Sigma^n \times \Sigma^n$ be a subset such that for any $(\vx,\ve)\in \good$, we have $F(\vx+\ve)=\vx$. 
        
        Let $\bad$ be the complement of $\good$, i.e., $\bad\defeq(\Sigma^n \times \Sigma^n)\setminus \good$.
        Suppose that we have 
        \begin{align}
        &\sum_{(\vx,\ve)\in \bad}|\hat{V}(\vx)\hat{W}(\ve)|^2\leq \epsilon \label{eq:hatV_hatW}\\
        &\sum_{\vz\in \Sigma^n}\left|\sum_{(\vx,\ve)\in \bad: \vx+\ve=\vz}\hat{V}(\vx)\hat{W}(\ve)\right|^2\leq \delta.
        \label{eq:hatV_hatW_two}
        \end{align}
         
        Let $U_{\mathsf{add}}$ and $U_{F}$ be unitaries defined as follows:
        \begin{align*}
        &\ket{\vx}\ket{\ve}
        \xrightarrow{U_{\mathsf{add}}}
        \ket{\vx}\ket{\vx+\ve}\xrightarrow{U_{F}}
        \ket{\vx-F(\vx+\ve)}\ket{\vx+\ve}.
        \end{align*}
        Then we have 
        \begin{align*}
           (I\otimes \QFT^{-1})U_{F}U_{\mathsf{add}}(\QFT\otimes \QFT)\ket{\psi}\ket{\phi}) \approx_{\sqrt{\epsilon}+\sqrt{\delta}} |\Sigma|^{n/2}\sum_{\vz \in \Sigma^n}(V\cdot W)(\vz)\ket{0}\ket{\vz}.
        \end{align*}

\end{lemma}

    Define set $T_i \coloneqq \{e \in \Sigma: H_i(e) = 0 \}$, $T \coloneqq T_1 \times T_2 \times \ldots \times T_n \in \Sigma^n$, one can represent
    \begin{align*}
    &\ket{\psi}=\sum_{\vx \in \Sigma^n}V(\vx)\ket{\vx}, \;
    \ket{\phi}=\sum_{\ve \in \Sigma^n}W(\ve)\ket{\ve}
    \end{align*}
    with
        \begin{align*}
        &V(\vx)=
        \begin{cases}
        \frac{1}{\sqrt{|C|}}& \vx\in C\\
        0& \text{otherwise}
        \end{cases}\\
        &W(\ve)=
        \begin{cases}
        \frac{1}{\sqrt{|T|}}& \ve\in T\\
        0& \text{otherwise}
        \end{cases}
    \end{align*}
    By definition, a string $z$ is a valid solution if and only if $V(z) \cdot W(z) > 0$. Therefore, by taking $F(\cdot) \coloneqq \decode(\cdot)$ and assuming  inequalities~\eqref{eq:hatV_hatW} and \eqref{eq:hatV_hatW_two} are satisfied, \autoref{lem: main technical} implies that Charlie will get a valid solution with high probability when he measures the second register of the resulted quantum state of the step (3).

    It remains to prove inequalities~\eqref{eq:hatV_hatW} and \eqref{eq:hatV_hatW_two} hold with high probability.

    \begin{claim}[Modified from Claim 6.4 in \cite{Yamakawa2022}]\label{claim 6.4}
        With high probability over $H = (H_1, \ldots, H_n)$ drawn from a $p$-biased distribution, there is a subset $\good\subseteq \Sigma^n \times \Sigma^n$ such that $\decode_{C^\perp}(\vx+\ve)=\vx$ for any $(\vx,\ve)\in \good$ and we have
        \begin{align*}
        &\sum_{(\vx,\ve)\in \bad}|\hat{V}(\vx)\hat{W}(\ve)|^2\leq \negl(n),\\
        &\sum_{\vz\in \Sigma^n}\left|\sum_{(\vx,\ve)\in \bad: \vx+\ve=\vz}\hat{V}(\vx)\hat{W}(\ve)\right|^2\leq\negl(n).
        \end{align*}
        where $\bad=(\Sigma^n \times \Sigma^n)\setminus \good$.
    \end{claim}

    To simplify \autoref{claim 6.4}, we reuse the definition of set $\gooderrors$ in \autoref{lem:code},

    \begin{equation*}
     \gooderrors \coloneqq \{e \in \Sigma^n: \hw(e)\leq (p+\varepsilon)n \}.  
    \end{equation*}

    Take $\baderrors \coloneqq \Sigma^n \backslash \gooderrors$. By \autoref{lem: unique decodability}, we can set $\good \coloneqq C^{\bot} \times \gooderrors$, and $\bad \coloneqq (\Sigma^n \times \Sigma^n) \backslash \good$. Plug them into \autoref{claim 6.4}, and note that $\hat{V}(x) = 0$ for all $x \notin C^{\bot}$, we have the following:

    \begin{align*}
        &\sum_{(\vx,\ve)\in \bad}|\hat{V}(\vx)\hat{W}(\ve)|^2
        =\sum_{\ve\in \baderrors}|\hat{W}(\ve)|^2,\\
        &\sum_{\vz\in \Sigma^n}\left|\sum_{(\vx,\ve)\in \bad: \vx+\ve=\vz}\hat{V}(\vx)\hat{W}(\ve)\right|^2
        =\sum_{\vz\in \Sigma^n}\left|\sum_{\substack{\vx\in C^\perp,\ve\in \baderrors\\: \vx+\ve=\vz}}\hat{V}(\vx)\hat{W}(\ve)\right|^2.
    \end{align*}

    Then, by an averaging argument, it suffices bound their expected values to be negligible, i.e., 

    \begin{align}
        &\Ex_{H}\left[\sum_{\ve\in \baderrors}|\hat{W}(\ve)|^2\right]\leq \negl(n),
        \label{eq:hatV_hatW_exp}\\
        &\Ex_{H}\left[\sum_{\vz\in \Sigma^n}\left|\sum_{\substack{\vx\in C^\perp,\ve\in \baderrors\\: \vx+\ve=\vz}}\hat{V}(\vx)\hat{W}(\ve)\right|^2\right]\leq \negl(n).
        \label{eq:hatV_hatW_two_exp}
    \end{align}

    We only prove inequality~\eqref{eq:hatV_hatW_exp} here, while the proof for inequality~\eqref{eq:hatV_hatW_two_exp} is the same as in \cite{Yamakawa2022}. 

    For any $i \in [n]$, let
    \[
        W_i(\ve_i)=
        \begin{cases}
        \frac{1}{\sqrt{|T_i|}}& \ve_i\in T_i\\
        0& \text{otherwise}
        \end{cases}
    \]
    We now prove the following claim.

    \begin{claim}[Modified from the first half of Claim 6.6 in \cite{Yamakawa2022}]\label{claim: 6.6 first}
        For all $i \in [n]$, it holds that
        \begin{align*}
        \Ex_{H_i}\left[|\hat{W}_i(\vzero)|^2\right] = 1- p, 
        \end{align*}
        
    \end{claim}

    \begin{proof}
        \[\Ex_{H_i}\left[|\hat{W}_i(\vzero)|^2\right]
    =\Ex_{H_i}\left[\left|\frac{1}{\sqrt{|\Sigma|}}\sum_{\vz\in \Sigma}W_i(\vz)\right|^2\right]
    =\frac{\Ex_{H_i}\left[|T_i|\right]}{|\Sigma|} = 1 - p.  \]
    \end{proof}

    The next claim follows from the same proof as in \cite{Yamakawa2022} by exploiting the symmetry.
    \begin{claim}[Second half of Claim 6.6 in \cite{Yamakawa2022}]\label{claim: 6.6 second}
        For all $i\in[n]$ and $\ve,\ve'\in \Sigma\setminus\{0\}$, it holds that
        \begin{align*}
        \Ex_{H_i}\left[|\hat{W}_i(\ve)|^2\right] =\Ex_{H_i}\left[|\hat{W}_i(\ve')|^2\right].
        \end{align*}
    \end{claim}

    Note that 
    \[\hat{W}^{H}(\ve)=\prod_{i=1}^{n}\hat{W}_i^{H_i}(\ve_i).\]

    Therefore, by combining \autoref{claim: 6.6 first} and \autoref{claim: 6.6 second}, we get that for all $\ve\in \Sigma^n$,
        \[\Ex_{H}\left[|\hat{W}(\ve)|^2\right] = \mathcal{D}_n(\ve),\]
    where $\mathcal{D}_n$ is the same distribution defined in \autoref{lem:code}---each symbol is $0$ with probability $1 - p$ and otherwise a uniformly random element of $\Sigma \setminus \{0\}$.

    Finally, by the linearity of expectation and \autoref{lem:code}, we conclude that 
    \[ \Ex_{H}\left[\sum_{\ve\in \baderrors}|\hat{W}(\ve)|^2\right] = \Pr_{\ve \sim \mathcal{D}_n}[\ve\in \baderrors]  \leq \negl(n).\]

\bigskip
\subsection*{Acknowledgments}
MG is supported by the Swiss State Secretariat for Education, Research and Innovation (SERI) under contract number MB22.00026.
TG is supported by ERC Starting Grant 101163189 and UKRI Future Leaders Fellowship MR/X023583/1. SJ and JL are supported by Scott Aaronson's Quantum Systems Accelerator grant.

JL thanks Yeyuan Chen and Zihan Zhang for discussion on list-recoverable codes. SJ thanks Dmytro Gavisnky for discussions on the history of quantum-classical communication separations.

\DeclareUrlCommand{\Doi}{\urlstyle{sf}}
\renewcommand{\path}[1]{\small\Doi{#1}}
\renewcommand{\url}[1]{\href{#1}{\small\Doi{#1}}}
\bibliographystyle{alphaurl}

\begin{thebibliography}{BRWY13}

\bibitem[AA14]{AA14}
Scott Aaronson and Andris Ambainis.
\newblock The need for structure in quantum speedups.
\newblock {\em Theory Comput.}, 10:133--166, 2014.
\newblock \href {https://doi.org/10.4086/toc.2014.v010a006} {\path{doi:10.4086/toc.2014.v010a006}}.

\bibitem[ABK{\etalchar{+}}21]{ABKST21}
Scott Aaronson, Shalev Ben{-}David, Robin Kothari, Shravas Rao, and Avishay Tal.
\newblock Degree vs. approximate degree and quantum implications of {H}uang's sensitivity theorem.
\newblock In {\em Proceedings of the 53rd Symposium on Theory of Computing (STOC)}, pages 1330--1342. ACM, 2021.
\newblock \href {https://doi.org/10.1145/3406325.3451047} {\path{doi:10.1145/3406325.3451047}}.

\bibitem[ABK24]{ABK24}
Scott Aaronson, Harry Buhrman, and William Kretschmer.
\newblock A qubit, a coin, and an advice string walk into a relational problem.
\newblock In {\em Proceedings of the 15th Innovations in Theoretical Computer Science Conference (ITCS)}, volume 287 of {\em LIPIcs}. Schloss Dagstuhl, 2024.
\newblock \href {https://doi.org/10.4230/lipics.itcs.2024.1} {\path{doi:10.4230/lipics.itcs.2024.1}}.

\bibitem[AKK16]{Ambainis16}
Andris Ambainis, Martins Kokainis, and Robin Kothari.
\newblock Nearly optimal separations between communication (or query) complexity and partitions.
\newblock In {\em Proceedings of the 31st Conference on Computational Complexity (CCC)}, volume~50 of {\em LIPIcs}, pages 4:1--4:14. Schloss Dagstuhl, 2016.
\newblock \href {https://doi.org/10.4230/LIPIcs.CCC.2016.4} {\path{doi:10.4230/LIPIcs.CCC.2016.4}}.

\bibitem[BBBV97]{BBBV}
Charles~H. Bennett, Ethan Bernstein, Gilles Brassard, and Umesh Vazirani.
\newblock Strengths and weaknesses of quantum computing.
\newblock {\em SIAM J. Comput.}, 26(5):1510--1523, 1997.
\newblock \href {https://doi.org/10.1137/S0097539796300933} {\path{doi:10.1137/S0097539796300933}}.

\bibitem[BBCR13]{BBCR13}
Boaz Barak, Mark Braverman, Xi~Chen, and Anup Rao.
\newblock How to compress interactive communication.
\newblock {\em SIAM Journal on Computing}, 42(3):1327--1363, 2013.
\newblock \href {https://doi.org/10.1137/100811969} {\path{doi:10.1137/100811969}}.

\bibitem[BCW99]{BCW98}
Harry Buhrman, Richard Cleve, and Avi Wigderson.
\newblock Quantum vs. classical communication and computation.
\newblock In {\em Proceedings of the 30th Symposium on Theory of Computing (STOC)}, pages 63--68. ACM, 1999.
\newblock \href {https://doi.org/10.1145/276698.276713} {\path{doi:10.1145/276698.276713}}.

\bibitem[BJK04]{Yossef04}
Ziv Bar{-}Yossef, T.~S. Jayram, and Iordanis Kerenidis.
\newblock Exponential separation of quantum and classical one-way communication complexity.
\newblock In {\em Proceedings of the 36th Symposium on Theory of Computing (STOC)}, pages 128--137. ACM, 2004.
\newblock \href {https://doi.org/10.1145/1007352.1007379} {\path{doi:10.1145/1007352.1007379}}.

\bibitem[BK24a]{BK24}
Shalev Ben{-}David and Srijita Kundu.
\newblock Oracle separation of {QMA} and {QCMA} with bounded adaptivity.
\newblock In {\em Proceedings of the 51st International Colloquium on Automata, Languages, and Programming (ICALP)}, volume 297 of {\em LIPIcs}, pages Art. No. 21, 18. Schloss Dagstuhl, 2024.
\newblock \href {https://doi.org/10.4230/lipics.icalp.2024.21} {\path{doi:10.4230/lipics.icalp.2024.21}}.

\bibitem[BK24b]{BenDavid24}
Shalev Ben{-}David and Srijita Kundu.
\newblock Separations in query complexity for total search problems.
\newblock Technical report, arXiv, 2024.
\newblock \href {https://doi.org/10.48550/ARXIV.2410.16245} {\path{doi:10.48550/ARXIV.2410.16245}}.

\bibitem[BRWY13]{braverman2013}
Mark Braverman, Anup Rao, Omri Weinstein, and Amir Yehudayoff.
\newblock Direct products in communication complexity.
\newblock In {\em Proceedings of the 54th Symposium on Foundations of Computer Science (FOCS)}, pages 746--755. IEEE, 2013.
\newblock \href {https://doi.org/10.1109/FOCS.2013.85} {\path{doi:10.1109/FOCS.2013.85}}.

\bibitem[BS21]{BS21}
Nikhil Bansal and Makrand Sinha.
\newblock $k$-forrelation optimally separates quantum and classical query complexity.
\newblock In {\em Proceedings of the 53rd Symposium on Theory of Computing (STOC)}, pages 1303--1316. ACM, 2021.
\newblock \href {https://doi.org/10.1145/3406325.3451040} {\path{doi:10.1145/3406325.3451040}}.

\bibitem[CCHL23]{Chen2023}
Sitan Chen, Jordan Cotler, Hsin-Yuan Huang, and Jerry Li.
\newblock The complexity of {NISQ}.
\newblock {\em Nature Communications}, 14(1), 2023.
\newblock \href {https://doi.org/10.1038/s41467-023-41217-6} {\path{doi:10.1038/s41467-023-41217-6}}.

\bibitem[CFK{\etalchar{+}}21]{CFKMP21}
Arkadev Chattopadhyay, Yuval Filmus, Sajin Koroth, Or~Meir, and Toniann Pitassi.
\newblock Query-to-communication lifting using low-discrepancy gadgets.
\newblock {\em SIAM Journal on Computing}, 50(1):171--210, 2021.
\newblock \href {https://doi.org/10.1137/19M1310153} {\path{doi:10.1137/19M1310153}}.

\bibitem[dRGR22]{Rezende2022}
Susanna de~Rezende, Mika G\"{o}\"{o}s, and Robert Robere.
\newblock Proofs, circuits, and communication.
\newblock {\em SIGACT News}, 53(1), 2022.
\newblock \href {https://doi.org/10.1145/3532737.3532745} {\path{doi:10.1145/3532737.3532745}}.

\bibitem[Gav16]{Gav16}
Dmitry Gavinsky.
\newblock Entangled simultaneity versus classical interactivity in communication complexity.
\newblock In {\em Proceedings of the 48th Symposium on Theory of Computing (STOC)}, pages 877--884. ACM, New York, 2016.
\newblock \href {https://doi.org/10.1145/2897518.2897545} {\path{doi:10.1145/2897518.2897545}}.

\bibitem[Gav19]{Gav19}
Dmitry Gavinsky.
\newblock Quantum versus classical simultaneity in communication complexity.
\newblock {\em IEEE Transactions on Information Theory}, 65(10):6466--6483, 2019.
\newblock \href {https://doi.org/10.1109/TIT.2019.2918453} {\path{doi:10.1109/TIT.2019.2918453}}.

\bibitem[Gav21]{Gavinsky21}
Dmitry Gavinsky.
\newblock Bare quantum simultaneity versus classical interactivity in communication complexity.
\newblock {\em IEEE Transactions on Information Theory}, 67(10):6583--6605, 2021.
\newblock \href {https://doi.org/10.1109/TIT.2021.3050528} {\path{doi:10.1109/TIT.2021.3050528}}.

\bibitem[GKK{\etalchar{+}}08]{GKK+08}
Dmitry Gavinsky, Julia Kempe, Iordanis Kerenidis, Ran Raz, and Ronald de~Wolf.
\newblock Exponential separation for one-way quantum communication complexity, with applications to cryptography.
\newblock {\em SIAM Journal on Computing}, 38(5):1695--1708, 2008.
\newblock \href {https://doi.org/10.1137/070706550} {\path{doi:10.1137/070706550}}.

\bibitem[GKRS19]{GKRS19}
Mika G{\"o}{\"o}s, Pritish Kamath, Robert Robere, and Dmitry Sokolov.
\newblock Adventures in monotone complexity and {TFNP}.
\newblock In {\em Proceedings of the 10th Innovations in Theoretical Computer Science Conference (ITCS)}, pages 38:1--38:19, 2019.
\newblock \href {https://doi.org/10.4230/LIPIcs.ITCS.2019.38} {\path{doi:10.4230/LIPIcs.ITCS.2019.38}}.

\bibitem[GLM{\etalchar{+}}16]{Goos2016}
Mika G{\"o}{\"o}s, Shachar Lovett, Raghu Meka, Thomas Watson, and David Zuckerman.
\newblock Rectangles are nonnegative juntas.
\newblock {\em SIAM Journal on Computing}, 45(5):1835--1869, 2016.
\newblock \href {https://doi.org/10.1137/15M103145X} {\path{doi:10.1137/15M103145X}}.

\bibitem[GPW20]{GPW20}
Mika G\"{o}\"{o}s, Toniann Pitassi, and Thomas Watson.
\newblock Query-to-communication lifting for {$\sf{BPP}$}.
\newblock {\em SIAM Journal on Computing}, 49(4):441--461, 2020.
\newblock \href {https://doi.org/10.1137/17M115339X} {\path{doi:10.1137/17M115339X}}.

\bibitem[GR08]{GR08}
Venkatesan Guruswami and Atri Rudra.
\newblock Explicit codes achieving list decoding capacity: Error-correction with optimal redundancy.
\newblock {\em IEEE Transactions on Information Theory}, 54(1):135--150, 2008.
\newblock \href {https://doi.org/10.1109/TIT.2007.911222} {\path{doi:10.1109/TIT.2007.911222}}.

\bibitem[GRT22]{GirishRT22}
Uma Girish, Ran Raz, and Avishay Tal.
\newblock Quantum versus randomized communication complexity, with efficient players.
\newblock {\em Computational Complexity}, 31(2):17, 2022.
\newblock \href {https://doi.org/10.1007/S00037-022-00232-7} {\path{doi:10.1007/S00037-022-00232-7}}.

\bibitem[GS99]{GS99}
Venkatesan Guruswami and Madhu Sudan.
\newblock Improved decoding of reed-solomon and algebraic-geometry codes.
\newblock {\em IEEE Transactions on Information Theory}, 45(6):1757--1767, 1999.
\newblock \href {https://doi.org/10.1109/18.782097} {\path{doi:10.1109/18.782097}}.

\bibitem[Huf52]{Huffman52}
David Huffman.
\newblock A method for the construction of minimum-redundancy codes.
\newblock {\em Proceedings of the IRE}, 40(9):1098--1101, 1952.
\newblock \href {https://doi.org/10.1109/JRPROC.1952.273898} {\path{doi:10.1109/JRPROC.1952.273898}}.

\bibitem[KN97]{Kushilevitz1997}
Eyal Kushilevitz and Noam Nisan.
\newblock {\em Communication Complexity}.
\newblock Cambridge University Press, 1997.
\newblock \href {https://doi.org/10.1017/CBO9780511574948} {\path{doi:10.1017/CBO9780511574948}}.

\bibitem[KR11]{KR11}
Bo'az Klartag and Oded Regev.
\newblock Quantum one-way communication can be exponentially stronger than classical communication.
\newblock In {\em Proceedings of the 43rd Symposium on Theory of Computing (STOC)}, pages 31--40. ACM, New York, 2011.
\newblock \href {https://doi.org/10.1145/1993636.1993642} {\path{doi:10.1145/1993636.1993642}}.

\bibitem[Liu23]{Liu23}
Qipeng Liu.
\newblock Non-uniformity and quantum advice in the quantum random oracle model.
\newblock In {\em Proceedings of the 42nd Conference on the Theory and Applications of Cryptographic Techniques (EUROCRYPT)}, volume 14004 of {\em Lecture Notes in Computer Science}, pages 117--143. Springer, 2023.
\newblock \href {https://doi.org/10.1007/978-3-031-30545-0\_5} {\path{doi:10.1007/978-3-031-30545-0\_5}}.

\bibitem[LLPY24]{li2024classical}
Xingjian Li, Qipeng Liu, Angelos Pelecanos, and Takashi Yamakawa.
\newblock Classical vs quantum advice and proofs under classically-accessible oracle.
\newblock In {\em Proceedings of the 15th Innovations in Theoretical Computer Science Conference (ITCS)}, volume 287:72 of {\em LIPIcs}, pages 1--19. Schloss Dagstuhl, 2024.
\newblock \href {https://doi.org/10.4230/lipics.itcs.2024.72} {\path{doi:10.4230/lipics.itcs.2024.72}}.

\bibitem[LSS09]{LSS09}
Troy Lee, Gideon Schechtman, and Adi Shraibman.
\newblock Lower bounds on quantum multiparty communication complexity.
\newblock In {\em Proceedings of the 24th Conference on Computational Complexity (CCC)}, pages 254--262. IEEE, 2009.
\newblock \href {https://doi.org/10.1109/CCC.2009.24} {\path{doi:10.1109/CCC.2009.24}}.

\bibitem[NC00]{NCbook}
Michael Nielsen and Isaac Chuang.
\newblock {\em Quantum computation and quantum information}.
\newblock Cambridge University Press, 2000.
\newblock \href {https://doi.org/10.1017/CBO9780511976667} {\path{doi:10.1017/CBO9780511976667}}.

\bibitem[Raz99]{Raz99}
Ran Raz.
\newblock Exponential separation of quantum and classical communication complexity.
\newblock In {\em Proceedings of the 31st Symposium on Theory of Computing (STOC)}, page 358–367. ACM, 1999.
\newblock \href {https://doi.org/10.1145/301250.301343} {\path{doi:10.1145/301250.301343}}.

\bibitem[Rud07]{rudra2007list}
Atri Rudra.
\newblock {\em List decoding and property testing of error-correcting codes}.
\newblock PhD thesis, University of Washington, 2007.

\bibitem[RY20]{Rao2020}
Anup Rao and Amir Yehudayoff.
\newblock {\em Communication Complexity: And Applications}.
\newblock Cambridge University Press, 2020.
\newblock \href {https://doi.org/10.1017/9781108671644} {\path{doi:10.1017/9781108671644}}.

\bibitem[SSW23]{SSW23}
Alexander Sherstov, Andrey Storozhenko, and Pei Wu.
\newblock An optimal separation of randomized and quantum query complexity.
\newblock {\em SIAM Journal on Computing}, 52(2):525--567, 2023.
\newblock \href {https://doi.org/10.1137/22M1468943} {\path{doi:10.1137/22M1468943}}.

\bibitem[Vad12]{Vadhan12psr}
Salil Vadhan.
\newblock Pseudorandomness.
\newblock {\em Foundations and Trends in Theoretical Computer Science}, 7(1-3):1--336, 2012.
\newblock \href {https://doi.org/10.1561/0400000010} {\path{doi:10.1561/0400000010}}.

\bibitem[WYZ23]{WYZ23}
Shuo Wang, Guangxu Yang, and Jiapeng Zhang.
\newblock Communication complexity of set-intersection problems and its applications.
\newblock Technical Report TR23-164, Electronic Colloquium on Computational Complexity (ECCC), 2023.
\newblock URL: \url{https://eccc.weizmann.ac.il/report/2023/164/}.

\bibitem[YZ24a]{Yamakawa2022}
Takashi Yamakawa and Mark Zhandry.
\newblock Verifiable quantum advantage without structure.
\newblock {\em Journal of the ACM}, 71(3):1--50, 2024.
\newblock \href {https://doi.org/10.1145/3658665} {\path{doi:10.1145/3658665}}.

\bibitem[YZ24b]{YZ23}
Guangxu Yang and Jiapeng Zhang.
\newblock Communication lower bounds for collision problems via density increment arguments.
\newblock In {\em Proceedings of the 56th Symposium on Theory of Computing (STOC)}, 2024.
\newblock \href {https://doi.org/10.1145/3618260.3649607} {\path{doi:10.1145/3618260.3649607}}.

\bibitem[Zha12]{Zhandry12}
Mark Zhandry.
\newblock Secure identity-based encryption in the quantum random oracle model.
\newblock In {\em Proceedings of the 32nd Cryptology Conference (CRYPTO)}, volume 7417 of {\em Lecture Notes in Computer Science}, pages 758--775. Springer, 2012.
\newblock \href {https://doi.org/10.1007/978-3-642-32009-5\_44} {\path{doi:10.1007/978-3-642-32009-5\_44}}.

\end{thebibliography}
\newcommand{\etalchar}[1]{$^{#1}$}

\end{document}